%% file: main.tex
\documentclass[acmsmall]{acmart}
\usepackage{amsmath,amsfonts}
\usepackage{algorithm}
\usepackage{array}
\usepackage{textcomp}
\usepackage{stfloats}
\usepackage{verbatim}
\usepackage{graphicx}
\usepackage{natbib}
\usepackage{microtype}
\input{math_commands.tex}

\usepackage{times}  % DO NOT CHANGE THIS
\usepackage{helvet}  % DO NOT CHANGE THIS
\usepackage{courier}  % DO NOT CHANGE THIS
\usepackage{graphicx} % DO NOT CHANGE THIS
\usepackage[utf8]{inputenc} % allow utf-8 input
\usepackage[T1]{fontenc}    % use 8-bit T1 fonts
\usepackage{hyperref}       % hyperlinks
\usepackage{booktabs}       % professional-quality tables
\usepackage{amsfonts}       % blackboard math symbols
\usepackage{nicefrac}       % compact symbols for 1/2, etc.
\usepackage{microtype}      % microtypography
\usepackage{xcolor}         % colors
\usepackage{mathrsfs}       % for mathscr
\usepackage{subfigure}
\usepackage{amsmath}        %for align
\usepackage{multirow}
\usepackage{algorithm}

\usepackage{mathtools}
\usepackage{amsthm}

\usepackage{microtype}
\usepackage{algpseudocode}
\usepackage{booktabs}
\usepackage{hyperref}

\usepackage[capitalize,noabbrev]{cleveref}

\theoremstyle{plain}
\newtheorem{theorem}{Theorem}[section]

\theoremstyle{definition}

\newtheorem{assumption}[theorem]{Assumption}
\theoremstyle{remark}

\usepackage{newfloat}
\usepackage{listings}

\allowdisplaybreaks 
\algnewcommand\algorithmicinput{\textbf{Input:}}
\algnewcommand\INPUT{\item[\algorithmicinput]}
\algnewcommand\algorithmicoutput{\textbf{Output:}}
\algnewcommand\OUTPUT{\item[\algorithmicoutput]}
\algnewcommand\algorithmicbegin{\textbf{begin}}
\algnewcommand\BEGIN{\item[\algorithmicbegin]}
\algnewcommand\algorithmicendbegin{\textbf{end}}
\algnewcommand\ENDBEGIN{\item[\algorithmicendbegin]}

\algdef{SE}[DOWHILE]{Do}{doWhile}{\algorithmicdo}[1]{\algorithmicwhile\ #1}%

\newcommand{\ma}[1]{{\color{black}{#1}}}
\newcommand{\m}[1]{{\color{black}{#1}}}
 
\newcommand{\TheName}{FedDHAD}
\newcommand{\TheHetName}{FedDH}
\newcommand{\TheDropoutName}{FedAD}
\newcommand{\TheHetNameE}{FedDHE}

\AtBeginDocument{%
  }

\begin{document}

\title{Efficient Federated Learning with Heterogeneous Data and Adaptive Dropout}

\author{Ji Liu}
\authornote{Corresponding author: jiliuwork@gmail.com}
\authornote{Both authors contributed equally to this research.}
\affiliation{%
  \institution{Baidu Research and Hithink RoyalFlush Information Network Co., Ltd.}
  \streetaddress{}
  \city{Hangzhou}
  \country{China}
}

\author{Beichen Ma}
\authornotemark[2]
\affiliation{%
  \institution{Baidu Research, China and Cornell University}
  \streetaddress{}
  \city{New York}
  \country{United States}
}

\author{Qiaolin Yu}
\affiliation{%
  \institution{Cornell University}
  \streetaddress{}
  \city{New York}
  \country{United States}
}

\author{Ruoming Jin}
\affiliation{%
  \institution{Kent State University}
  \streetaddress{}
  \city{Kent, Ohio}
  \country{United States}
}

\author{Jingbo Zhou}
\affiliation{%
  \institution{BIL, Baidu Research}
  \streetaddress{}
  \city{Beijing}
  \country{China}
}

\author{Yang Zhou}
\affiliation{%
  \institution{Department of Computer Science and Software Engineering, Auburn University}
  \streetaddress{}
  \city{Auburn}
  \country{United States}
}

\author{Huaiyu Dai}
\affiliation{%
  \institution{Department of Electrical and Computer Engineering, North Carolina State University}
  \streetaddress{}
  \city{North Carolina}
  \country{United States}
}

\author{Haixun Wang}
\affiliation{%
  \institution{Instacart}
  \streetaddress{}
  \city{San Francisco, California}
  \country{United States}
}

\author{Dejing Dou}
\affiliation{%
  \institution{BEDI Cloud and School of Computer Science, Fudan University}
  \streetaddress{}
  \city{Beijing}
  \country{China}
}

\author{Patrick Valduriez}
\affiliation{%
  \institution{Inria, University of Montpellier, CNRS, LIRMM,  France and LNCC}
  \streetaddress{}
  \city{Petropolis, Rio de Janeiro}
  \country{Brazil}
}

\begin{abstract}
\label{sec:abs}
Federated Learning (FL) is a promising distributed machine learning approach that enables collaborative training of a global model using multiple edge devices.
The data distributed among the edge devices is highly heterogeneous. Thus, FL faces the challenge of data distribution and heterogeneity, where non-Independent and Identically Distributed (non-IID) data across edge devices may yield in significant accuracy drop. 
Furthermore, the limited computation and communication capabilities of edge devices increase the likelihood of stragglers, thus leading to slow model convergence. In this paper, we propose the \TheName{} FL framework, which comes with two novel methods: Dynamic Heterogeneous model aggregation (\TheHetName{}) and Adaptive Dropout (\TheDropoutName{}). \TheHetName{} dynamically adjusts the weights of each local model within the model aggregation process based on the non-IID degree of heterogeneous data to deal with the statistical data heterogeneity.
\TheDropoutName{} performs neuron-adaptive operations in response to heterogeneous devices to improve accuracy while achieving superb efficiency. The combination of these two methods makes \TheName{} significantly outperform state-of-the-art solutions in terms of accuracy (up to 6.7\% higher), efficiency (up to 2.02 times faster), and computation cost (up to 15.0\% smaller). 
\end{abstract}

%\begin{IEEEkeywords}
%\end{IEEEkeywords}

\received{21 June 2024}
\received[revised]{21 June 2024}
\received[accepted]{21 June 2024}

\maketitle

\section{Introduction}
\label{sec:intro}
Ever growing, high numbers of edge devices (devices for short) generate huge amounts of distributed data that can be very useful for machine learning.
To avoid the common problems of a centralized approach (high data transfer cost, long training time, ...),
Federated Learning (FL) \cite{mcmahan2021advances,chen2024trustworthy,liu2024enhancing} collaboratively trains a global model using multiple edge devices.
While the data generally contain sensitive information about end-users, e.g., facial images, location-based services, and account information \cite{liu2022distributed}, 
multiple legal restrictions \cite{GDPR,CCPA} protect the privacy and security of distributed data, so that data aggregation from distributed devices is almost impossible. % \cite{yang2019federated}. 
Within a typical FL architecture \cite{mcmahan2017communication}, a powerful parameter server module (server for short) \cite{liu2023heterps,li2014scaling} 
coordinates the training process of multiple devices
%with heterogeneous modest computation and communication capacity, which
and aggregates the trained models instead of the raw data from the devices, thus providing security and privacy.  

As multiple devices may belong to various end-users, the data that is distributed among the edge devices is highly heterogeneous \cite{Wang2020Tackling,liu2022multi,liu2022Efficient} in terms of distribution, \textit{i.e.}, non-Independent and Identically Distributed (non-IID) \cite{mcmahan2017communication}. While data format or semantics are generally homogeneous, the statistical data distribution among devices is highly heterogeneous \cite{LiStudy2022}. For instance, the proportion of the samples corresponding to Class A on a device can be very different from that of other devices. While FL was proposed to train a global model using non-IID data on mobile devices \cite{mcmahan2017communication,Li2020On}, the traditional FL models may incur inferior accuracy \cite{Li2020FedProx} %,
%,zhao2018federated,li2021federated}, 
%due to inconsistent objectives \cite{Wang2020Tackling} or client drift \cite{Karimireddy2020SCAFFOLD} 
brought by such statistical data heterogeneity. The training process of FL generally corresponds to multiple rounds, each with three steps. In Step 1, the server sends the global model to selected devices. In Step 2, the selected devices update the global model with local data and upload the updated model to the server. In Step 3, the server performs the model aggregation with the uploaded models. Conventional model aggregation methods, \textit{e.g.}, FedAvg \cite{mcmahan2017communication}, ignore the property of non-IID data while aggregating the local models, which would lower accuracy. 

Recent techniques have improved the performance of FL on non-IID data using additional proximal terms (Step 2) in the local update, \textit{e.g.}, FedProx \cite{Li2020FedProx}, or based on gradient normalization and averaged steps in the model aggregation (Step 3), \textit{e.g.}, FedNova \cite{Wang2020Tackling}. 
However,
the trained model still suffers from data heterogeneity. The performance of the trained model may be unacceptably poor, and the training process may even become unstable in practice. We quantify the heterogeneity between the distribution of the dataset on a device and that of all the datasets on all available devices by non-IID degree using Formula \ref{eq:non-iid-def} (see details in Section \ref{subsec:theoreticalAna}).
Although the Jensen–Shannon (JS) divergence~\cite{fuglede2004jensen} or Kullback-Leibler (KL) divergence \cite{kullback1997information} can be exploited, % to represent the non-IID degree, 
there is still a difference between JS and KL divergences and the non-IID degree. %In this paper, we utilize the non-IID degree to dynamically adjust the weights to alleviate the statistical heterogeneity issue brought by the heterogeneous data 
%so as to improve the model performance within an FL setting. 

To alleviate the statistical data heterogeneity issue, a promising solution is
to adapt for each model the weight tuning to the non-IID degree.
The global model optimization with weight tuning can be formulated as a bi-level optimization problem \cite{Xie2021Multi},
%Bard1998PracticalBO, Xie2021Multi}, 
where the global model optimization and the weight adjustment belong to the optimization of two levels. However, there are two major difficulties in finding the optimal weights of each model, which is critical for the training process of FL. First, the training process dynamically varies and the optimal weights may vary in various epochs. Second, it is difficult to accurately calculate the non-IID degrees, which are critical for the optimal weights.

While the server is generally powerful, some devices typically have modest computation and communication capacities \cite{mcmahan2017communication}. For instance, in a real phone keyboard application \cite{Bonawitz2019Towards},
6-10\% devices may take an unaccepted long time to return the updated models.
Such device heterogeneity leads to slow model convergence due to low efficiency of the local training process and unbalanced workload. High-dimensional model parameters (\textit{e.g.}, deep neural networks) in the FL paradigm give rise to extremely high computation and communication overhead. Within a standard FL environment, the bandwidth between the devices and server is generally of low quality. Furthermore, models become bigger and bigger, especially for Large Language Models (LLM), to achieve superb performance, which incurs significant computation and communication costs. Thus, it takes much time to send the global model from the server to the devices or to upload the updated models from each device to the server.
%\cite{caldas2018expanding}. 
In addition, heterogeneous hardware resources across devices increase the likelihood of stragglers, i.e., the slowest participating devices, which incur an unbalanced workload.
%In addition, as the computation and communication capacity of various devices significantly differ, some powerful selected devices may need to wait for stragglers, \textit{i.e.}, the devices that take more time to perform Step 2 than the powerful devices. 

Dropout techniques \cite{Srivastava2014Dropout} randomly remove neurons during training to reduce overfitting and to achieve superior performance.
They have been exploited in FL for reducing the size of the global model and thus communication and computation cost \cite{horvath2021fjord}.
However, the straightforward application of dropout to FL may incur significant accuracy degradation of the global model \cite{Bouacida2021FedDropout} with improper dropout rates \cite{wen2022federated}. In order to fully exploit the heterogeneous devices, adaptive pruning rates should be utilized. 
Furthermore, model aggregation with dropout becomes even more challenging than without dropout as it is hard to verify the importance of the updated models with some neurons removed. Thus, it is critical to dynamically adjust the weights with the adaptive pruning operation.

Statistical data heterogeneity and device heterogeneity are often regarded as orthogonal problems, which are addressed separately. However, solving either problem leads to inaccurate and inefficient analytical results.
Some methods can be exploited to address the non-IID data issue, \textit{e.g.}, regularization-based method \cite{Li2020FedProx,acar2021federated}, gradient normalization \cite{Wang2020Tackling}, contrastive learning \cite{li2021model}, client drift adjustment \cite{Karimireddy2020SCAFFOLD}, %momentum-based method \cite{hsu2019measuring}, 
%classifier calibration \cite{Luo2021No}, 
and personalization \cite{Sun2021PartialFed}.
%ozkara2021quped,Sun2021PartialFed}. 
Some other methods have been addressed the device heterogeneity problem, \textit{e.g.}, asynchronous staleness aggregation \cite{park2021sageflow,jia2024efficient,liu2024aedfl,jia2025efficient,liu2024fedasmu}. %, heterogeneous models \cite{diao2021heterofl}. 
However, \m{some of them} can be combined with our approach to achieve superior performance. 

% In this paper, we propose the \TheName{} framework to simultaneously address statistical data heterogeneity and device heterogeneity in FL with two novel methods: Dynamic Heterogeneous model aggregation (\TheHetName{}) and Adaptive Dropout (\TheDropoutName{}).
% \TheHetName{} exploits the estimated non-IID degrees of devices to dynamically adjust the weights of each local model within the model aggregation process. 
% \TheDropoutName{} is a neuron-adaptive model dropout method that adapts the dropout operation to each device while reducing communication and computation costs and achieving an excellent balance between efficiency and effectiveness.
% While \TheDropoutName{} can compensate the slight overhead brought by \TheHetName{}, we integrate these two methods  into \TheName{}, thus accelerating the training process and improving accuracy.
% We conduct extensive experiments to compare \TheName{} with representative approaches with two typical models over two real-world datasets. Experimental results demonstrate the major advantages of \TheName{} in terms of accuracy, efficiency, and computation cost.

\ma{In this paper, we propose the \TheName{} framework to simultaneously address the data heterogeneity and the heterogeneous limited capacity problem. We propose exploiting the non-IID degrees of devices to dynamically adjust the weights of each local model within the model aggregation process, \textit{i.e.}, \TheHetName{}, in order to alleviate the impact of the heterogeneous data. %Moreover, FedDH can be compatible with existing FL algorithms that adds proximal terms to local update or uses gradient normalization \cite{Li2020FedProx, Wang2020Tackling}, and further improves upon them. 
Then, we introduce a neuron-adaptive model dropout method, \textit{i.e.}, \TheDropoutName{}, to perform the dropout operation adapted to each device while achieving load balancing so as to reduce the communication and computation cost and to achieve a good balance between efficiency and effectiveness. While \TheDropoutName{} can compensate the slight overhead brought by \TheHetName{}, we integrate these two methods  into \TheName{}, thus accelerating the training process and improving accuracy.
We conduct extensive experiments to compare \TheName{} with representative approaches with two typical models over two real-world datasets. Experimental results demonstrate the major advantages of \TheName{} in terms of accuracy, efficiency, and computation cost. The major contributions are summarized as follows:
\begin{itemize}
    \item We propose a Dynamic Heterogeneous model aggregation algorithm for non-IID data, \textit{i.e.}, \TheHetName{}. \TheHetName{} takes advantage of the non-IID degree of devices and dynamically adjusts the weights of local models within the model aggregation process.
    \item We propose an Adaptive Dropout method, \textit{i.e.}, \TheDropoutName{}. \TheDropoutName{} carries out dropout operations adapted to each neuron and diverse idiosyncrasies, \textit{i.e.}, heterogeneous communication and computation capacity, of each device to improve the efficiency of the training process.
    \item We conduct extensive experiments to compare \TheName{} with representative approaches with two typical models over two real-world datasets. Experimental results demonstrate the superb advantages of \TheName{} in terms of accuracy, efficiency, and computation cost.
\end{itemize}}

The rest of this paper is organized as follows. Section \ref{sec:preliminaries} provides some background on FL. \ma{Section \ref{sec:relatedwork} describes related work.} Section \ref{sec:sysModel} introduces the system model of \TheName{}. In Section \ref{sec:method}, we describe the \TheName{} framework, including the \TheHetName{} and \TheDropoutName{} methods. Section \ref{sec:experiments}, gives our experimental evaluation. Finally, Section \ref{sec:conclusion} concludes.
% \section*{Appendix A.}
% \label{app:theorem}

% % Note: in this sample, the section number is hard-coded in. Following
% % proper LaTeX conventions, it should properly be coded as a reference:

% %In this appendix we prove the following theorem from
% %Section~\ref{sec:textree-generalization}:

% In this appendix we prove the following theorem from
% Section~6.2:

% \noindent
% {\bf Theorem} {\it Let $u,v,w$ be discrete variables such that $v, w$ do
% not co-occur with $u$ (i.e., $u\neq0\;\Rightarrow \;v=w=0$ in a given
% dataset $\dataset$). Let $N_{v0},N_{w0}$ be the number of data points for
% which $v=0, w=0$ respectively, and let $I_{uv},I_{uw}$ be the
% respective empirical mutual information values based on the sample
% $\dataset$. Then
% \[
%   N_{v0} \;>\; N_{w0}\;\;\Rightarrow\;\;I_{uv} \;\leq\;I_{uw}
% \]
% with equality only if $u$ is identically 0.} \hfill\BlackBox

% \noindent
% {\bf Proof}. We use the notation:
% \[
% P_v(i) \;=\;\frac{N_v^i}{N},\;\;\;i \neq 0;\;\;\;
% P_{v0}\;\equiv\;P_v(0)\; = \;1 - \sum_{i\neq 0}P_v(i).
% \]
% These values represent the (empirical) probabilities of $v$
% taking value $i\neq 0$ and 0 respectively.  Entropies will be denoted
% by $H$. We aim to show that $\fracpartial{I_{uv}}{P_{v0}} < 0$....\\

% {\noindent \em Remainder omitted in this sample. See http://www.jmlr.org/papers/ for full paper.}

% \vskip 0.2in
% \bibliography{ref}

\section{Background on Federal Learning}
\label{sec:preliminaries}

In this section, we introduce some background concepts on FL. For simplicity, we use the notations summarized in Table \ref{tab:summary}. 

\textbf{Taxonomy.} FL can be classified into cross-silo and cross-device, in terms of participants \cite{mcmahan2021advances}. The participants of cross-device FL are generally end-users with edge devices, e.g., mobiles. The number of participants is about 100 or bigger. The end-user is the person who generates the raw data for FL. The participants of cross-silo FL generally consist of multiple organizations and the number of participants is generally less than 100. In addition, FL can be further divided according to the kind of data partitioning, i.e., horizontal, vertical, and hybrid \cite{yang2019federated}. When the samples from each participant have the same dimensions with the same semantics but correspond to various end-users, the setting is horizontal FL. When the samples of each participant are of different dimensions (various semantics) but correspond to the same end-users, the setting is vertical FL. Otherwise, when the samples of each participants are of different dimensions and correspond to various end-users, the setting is hybrid FL. In this paper, we focus on horizontal cross-device FL. \ma{Recent work has explored one-shot semi-supervised FL using pre-trained diffusion models for category-based classification \cite{yang2024exploring}. This approach leverages the power of diffusion models to achieve learning with limited labeled data in a federated setting.} In addition, we conduct multiple rounds of training in our system, which deviates from one-shot federated learning \cite{SalehkaleybarOneShot}. 

\textbf{Objectives of FL.} In a horizontal cross-device FL environment composed of a powerful server and $N$ devices, we assume that the raw data is distributed at each device and cannot be transferred while there is no raw data on the server for training. We assume that Dataset $D_k = \{x_{k}, y_{k}\}^{n_k}$ is located at Device $k$, consisting of $n_k$ samples. $x_{k}$ represents the input data sample on Device $k$, and $y_{k}$ is the corresponding label. We denote the number of samples on all the devices by $n$. Then, the objective of the FL training process can be formulated as:
\vspace{-2mm}
\begin{equation}
\vspace{-2mm}
\min_{w}\left[F(w)\triangleq\frac{1}{n}\sum_{k = 1}^N n_k F_k(w)\right],
\label{eq:problem}
\end{equation}
where $w$ is the parameters of the global model, $F_k(w)$ is the local loss function of Device $k$ with $f(w,x_{k},y_{k})$ capturing the error of the model on the sample $\{x_{k},y_{k}\}$, as defined in Formula \ref{eq:fk}.
\begin{equation}
\vspace{-2mm}
F_k(w)\triangleq\frac{1}{n_k}\sum_{\{x_{k},y_{k}\} \in \mathcal{D}_k} f(w,x_{k},y_{k})
\label{eq:fk}
\end{equation}

\begin{table}[t]
\caption{Summary of main notations.}
% \vspace{-5mm}
\begin{center}
\begin{tabular}{cc}
\toprule
Notation & Definition \\
\hline
\label{tab:summary}
%\cite{low2010optimized}
$\mathcal{K}$ & Set of all devices \\
$t$ & The number of global round \\
$S_t$ & The set of selected devices \\
$N$ & The size of $\mathcal{K}$ \\
$\mathcal{D}_k$ & Local dataset on Device $k$\\
$n_k$ & The size of $\mathcal{D}_k$ \\
$\mathcal{D}$; $n$ & Global dataset; size of $D$ \\
$d_k$ & Batch size of the local update of Device $k$ \\
$F$& Global loss function \\
$F_k$ & Local loss function on Device $k$\\
$w^*$ & The optimal model \\
$w^k_{t,h}$ & Local model on Device $k$, Local update $h$ \\
$\bar{w}_{t,h}$ & Global model in the $h$-th local update \\
$\eta_{t,h}$ & Learning rate in the $h$-th local update \\
$\bar{g}_{t,h}$ & Global gradients in the $h$-th local update \\
$\tau$ & Number of local epochs\\
%$C$ & The ratio between the number of devices scheduled to $|\mathcal{K}|$\\
$\zeta^{k}_{t,h}$ & The sampled dataset at local iteration $h$ on Device $k$\\
$|\zeta^{k}_{t,h}|$ & The size of $\zeta^{k}_{t,h}$\\
%$\nabla f_k(w; \zeta^{k}_{t,h})$ & Stochastic gradient with $\zeta^{k}_{t,h}$ on Device $k$\\
%$\mathscr{C}$ & Number of rounds to avoid accuracy degradation in critical learning periods\\
$p_{k,t}$ & The weights in the model aggregation for Device $k$ \\ %with $\sum_{k \in S_t}p_{k,t} = 1$\\
\bottomrule
\vspace{-6mm}
\end{tabular}
\end{center}
\end{table}

\textbf{Data heterogeneity.} The data is generally highly heterogeneous as the distribution can be both non-identical and non-independent, i.e., non-IID. Let us denote the possibility to get a sample from sampling the local data of Device $i$ by $\mathcal{P}_i$. Then, the distribution is non-identical when $\mathcal{P}_i \neq \mathcal{P}_j$, where $i \neq j$. The data of each device is non-independent as the end-users may have dependencies on geographic relationships \cite{eichner2019semi}. For instance, the data on the end-users of the same region may have similar patterns. Furthermore, the number of samples on each device can be significantly different.

\textbf{Challenges.} As a \textit{de facto} standard FL approach, FedAvg \cite{mcmahan2017communication} is a simple yet effective classic approach for FL.
%that addresses two major challenges, i.e., data heterogeneity and device heterogeneity. 
With FedAvg, the training process of FL contains three steps in each round as presented in Section \ref{sec:intro}. In Step 3, FedAvg calculates a simple average model with the updated local models. While the data is highly heterogeneous, the data that severely deviate the global distribution may have important impact on the global model with FedAvg, which may significantly degrade the accuracy of the aggregated model \cite{Li2020FedProx}. In addition, the devices are generally of heterogeneous computational and communication capacity. FedAvg transfers the whole model between devices and the server, which incurs significant computational and communication costs and unbalanced workload. Some devices may take an unaccepted time to upload the updated local model, which is inefficient. Thus, device heterogeneity increases the total training time in FL. In this paper, we propose a novel framework, i.e., \TheName{}, to address these two challenges.

\section{Related Work}
\label{sec:relatedwork}

In this section, we introduce \m{the existing FL methods dealing with the model performance issue incurred by data heterogeneity. Then, we} present related methods to handle the training efficiency problem.

\begin{table}[!t]
\caption{Server computation overhead}
% \vspace{-4mm}
\label{tab:server_burden}
\begin{center}
\begin{tabular}{l|ll|ll}
\hline
        & \multicolumn{2}{c|}{CIFAR10}       & \multicolumn{2}{c}{CIFAR100}       \\ \cline{2-5} 
        & \multicolumn{1}{l|}{LeNet} & CNN   & \multicolumn{1}{l|}{LeNet} & CNN   \\ \hline
FedDHAD & \multicolumn{1}{l|}{0.122} & 0.156 & \multicolumn{1}{l|}{0.126} & 0.145 \\ \hline
FedDH   & \multicolumn{1}{l|}{0.086} & 0.042 & \multicolumn{1}{l|}{0.104} & 0.043 \\ \hline
\end{tabular}
\vspace{-7mm}
\end{center}
\end{table}

Conventional methods, e.g., FedAvg \cite{mcmahan2017communication}, can converge with non-IID data \cite{Li2020On}, while they incur modest accuracy due to inconsistent objectives \cite{Wang2020Tackling} or client drift \cite{Karimireddy2020SCAFFOLD}. Several methods have been proposed to alleviate the problem. FedProx \cite{Li2020FedProx} utilizes regularization in the local objective to tolerate the difference between the local model and the global model. \m{FedCurv \cite{casella2023benchmarking} calculate global parameter stiffness to control discrepancies.} FedDyn \cite{acar2021federated} further introduces dynamic regularization for each device. %Scaffold \cite{Karimireddy2020SCAFFOLD} adjusts the direction of each device based on variance reduction. 
FedNova \cite{Wang2020Tackling} exploits gradient normalization to reduce the influence brought by the heterogeneous learning steps with non-IID data. MOON \cite{li2021model} exploits contrastive learning to adjust the local training based on the similarity of model representations. \m{The primary limitation of these methods is that they concentrate on performance within a single domain under label skew conditions, neglecting the issue of domain shift, resulting in unsatisfactory performance across multiple domains. FedHEAL \cite{chen2024fair} reduces parameter update conflicts by discarding unimportant parameters and implements fair aggregation to prevent domain bias. However, it requires careful hyperparameter tuning and assumes homogeneous network architectures across all clients, limiting its applicability in heterogeneous settings.} %CCVR \cite{Luo2021No} calibrates the classifier using the virtual representations based on a mixture model. 
%FedLin \cite{Mitra2021FedLin} exploits past gradients to address the objective heterogeneity brought by non-IID data. 
Momentum-based methods %\cite{hsu2019measuring} 
deal with non-IID data using momentum at the server side %\cite{reddi2021adaptive}, 
or at the device side \cite{JinAccelerated2022}.% \cite{karimireddy2020mime}. %, and at both sides . 
Various methods, \textit{e.g.}, partial personalized model aggregation in PartialFed \cite{Sun2021PartialFed}, %personalized quantization in QuPeD \cite{ozkara2021quped}, %device clustering in CFL \cite{sattler2020clustered}, 
meta-learning-based method \cite{Khodak2019Adaptive}, 
and multi-task learning \cite{Smith2017Federated}, are proposed to realize personalization for non-IID data. \m{ADCOL \cite{li2023adversarial} and FCCL \cite{huang2023generalizable} emphasize personalized model training instead of optimizing a single shared global model. However, these approaches necessitate additional components such as discriminators and access to public datasets, imposing significant computational and resource burdens on both participating devices and the central server.} Several other methods \cite{Yurochkin2019Bayesian}
%,Wang2020Federated,Yu2021Fed2} 
align the models within the model aggregation step in order to avoid structural misalignment incurred by the permutation of neural network parameters. 
In addition, some other works \cite{Lin2020Ensemble}
%,jeong2018communication,he2020group,yoshida2020hybrid,zhao2018federated, nagalapatti2021game, Zhang2022FedDUAP} 
exploit the publicly available data in the server while handling the non-IID, which may incur security or privacy issues.  \m{Recent advancements in large pretrained models have led to efforts to incorporate them into FL. FedKTL \cite{zhang2024upload} introduced a knowledge-distillation-based method, that transfer distilled prototypes or representations to clients, addressing the challenges in knowledge transfer caused by data and model heterogeneity. In addition, AugFL \cite{yue2025augfl} devised an inexact-ADMM-based algorithm for pretrained model personalized FL, allowing knowledge transfer from a private pretrained model at the server to clients while reducing computational cost. However, sharing pretrained model parameters or representations poses risks to the ownership of the these parameters \cite{kang2023grounding}. Furthermore, transferring additional parameters or representations to clients may increase their computational and storage requirements.}
While bi-level optimization is utilized to train federated graph models with linear transformation layers and graph convolutional networks in BiG-Fed \cite{Xing2021BiG-Fed}, this approach does not consider the non-IID degrees of the distributed data to future improve the accuracy of FL.

Finally, multiple methods address the heterogeneous computation and communication capacity of devices in FL. As an efficient method to reduce overfitting and achieve excellent performance, dropout \cite{Srivastava2014Dropout} is exploited to reduce the communication and computation costs. \m{Federated Dropout was proposed by extending this approach, which reduces communication costs by deriving small sub-models from the global model for local updates and exchanging only these sub-models between the servre and clients.} However, dropout may incur accuracy degradation with an improper dropout rate \cite{wen2022federated}. While empirical experience may help set the dropout rate,%\cite{wen2022federated},
%,caldas2018expanding}, 
the diversity of models and datasets brings difficulties in FL. Furthermore, a simple assessment of the dropout strategy based on the variation of loss values \cite{Bouacida2021FedDropout} does not consider the nature loss reduction within the FL training process and may correspond to low efficiency or low accuracy. \m{Adaptive Federated Dropout maintains an activation score map to generate suitable sub-models for each client, enabling participation from clients with limited capabilities. However, it falls short in providing customized pruned sub-models tailored to different clients' specific needs.} An ordered dropout mechanism is introduced to achieve ordered and nested representations of multiple models \cite{horvath2021fjord}, which exploits a dropout rate space. However, the dropout rate space is still difficult to construct. %Quantization \cite{reisizadeh2020fedpaq}
%,caldas2018expanding,konevcny2016federated} 
%and 
Pruning 
methods are exploited to reduce the communication and computation costs \cite{jiang2022model,bibikar2022federated}, which incur lossy operations and may degrade the accuracy. 

%\ma{Our proposed framework, \TheName{}, distinguishes from traditional federated learning approaches by addressing both data heterogeneity and the issue of limited capacity in heterogeneous devices simultaneously, unlike existing works that focus on either data or device heterogeneity in isolation. We propose exploiting the non-IID degrees of devices to dynamically adjust the weights of each local model within the model aggregation process, \textit{i.e.}, \TheHetName{}, in order to alleviate the impact of the heterogeneous data. Furthermore, we introduce a neuron-adaptive model dropout method, \textit{i.e.}, \TheDropoutName{}, to perform the dropout operation adapted to each device while achieving load balancing so as to reduce the communication and computation cost and to achieve a good balance between efficiency and effectiveness. While \TheDropoutName{} can compensate the slight overhead brought by \TheHetName{}, we integrate these two methods into \TheName{}, achieving superior accuracy, efficiency, and computational cost reductions.}

\section{System Model \& Problem Formulation}
\label{sec:sysModel}

In this section, we present our system model and formally formulate the problem we address. 

\textbf{System model.} During the training process of \TheName{}, the global model is updated with multiple rounds. As shown in Figure \ref{fig:framework}, each round is composed of five steps. Within each round, the server randomly selects $N' = N * C$ devices, where $C$ represents the ratio between the number of the selected devices and $N$. In Step \textcircled{1}, dropout operations are carried out and $N'$ sub networks (subnet) are generated \ma{(see details in Section \ref{subsec:fedAD})}. The dropout rates of each device are updated when certain conditions are met (see details in Section \ref{subsec:fedAD}). Then, the sub neural networks are downloaded to the selected devices \ma{through receiving the current global model parameters from the server} in Step \textcircled{2}. Afterward, 
% the local model is updated with the local dataset
\ma{each selected devices then performs local computation based on the local dataset (see details in Section \ref{subsec:fedDH})} in Step \textcircled{3}, and the updated model is uploaded to the server in Step \textcircled{4}. Finally, the updated models are aggregated with the consideration of the data heterogeneity in Step \textcircled{5}. While the model aggregation with the consideration of the data heterogeneity (see details in Section \ref{subsec:fedDH}) introduces slight computation overheads, we exploit adaptive dropout (see details in Section \ref{subsec:fedAD}) to accelerate the training process. 

\begin{figure}[!t]
%\vspace{-2mm}
\centering
\includegraphics[width=0.8\linewidth]{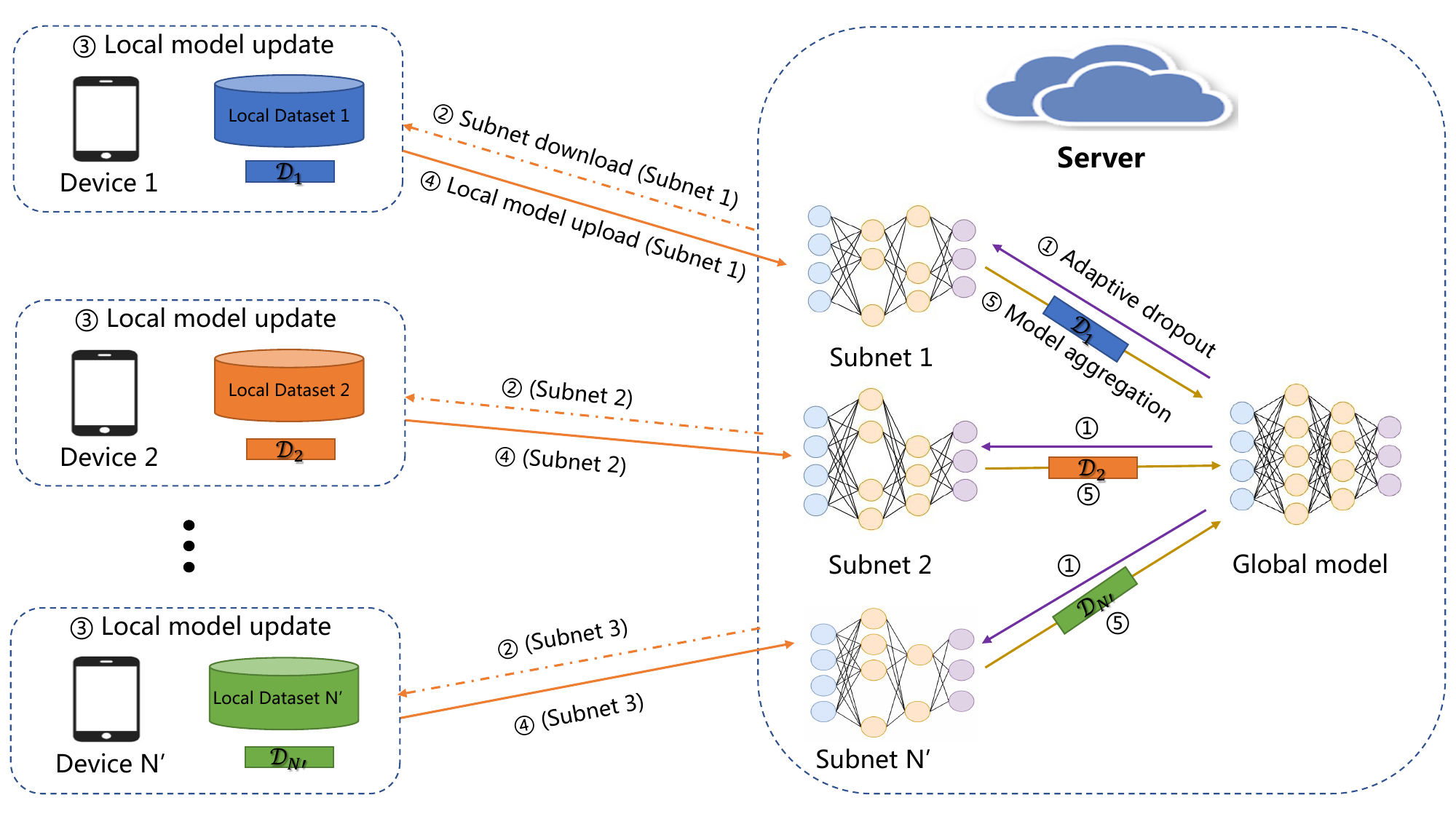}
\vspace{-4mm}
\caption{The system model of \TheName{}.}
\label{fig:framework}
\vspace{-6mm}
\end{figure}

\textbf{Problem formulation.} The problem to address in this paper is how to efficiently train a global model of high performance with a FL setting. We break down the problem into two aspects, i.e., model performance and training efficiency. Thus, the first problem is how to achieve high model performance (defined in Formula \ref{eq:problem}) while considering the data heterogeneity. As data heterogeneity is critical to the model performance, the model performance problem consists of how to estimate non-IID degrees and how to leverage the non-IID degree to dynamically adjust the importance (weight) of each model. We formulate the training efficiency problem as how to minimize the training time of each round with heterogeneous devices while ensuring the performance of the trained model.

\section{Efficient Heterogeneous FL with \TheName{}}
\label{sec:method}

%Ji: this sentence is confusing: where is the algorithm of \TheName{}? BTY, so far, you said \TheName{} is a framework. Now, it is an algorithm?
%JL: I changed the word "algorithm" to "framework".

In this section, we present the framework of \TheName{}.
We first
propose the dynamic heterogeneous model aggregation method, \textit{i.e.}, \TheHetName{}. Then, we present the adaptive dropout method, \textit{i.e.}, \TheDropoutName{}, to reduce the communication and computation costs in order to accelerate the training process while achieving excellent accuracy with heterogeneous devices. Then, we present the integration of the dynamic heterogeneous model aggregation method, \textit{i.e.}, \TheHetName{} and the adaptive dropout method, \textit{i.e.}, \TheDropoutName{}.

\subsection{Heterogeneous Model Aggregation}
\label{subsec:fedDH}

In this section, we propose our dynamic heterogeneous federated learning method, \textit{i.e.}, \TheHetName{}, which exploits the non-IID degrees of the heterogeneous data on each device to adjust the weights of multiple models in order to address the model performance problem. We first introduce a dynamic method to estimate the non-IID degree. Then, we present the details of \TheHetName{}. Afterward, we provide the theoretical convergence analysis of \TheHetName{}. 

In order to address the model performance problem, we reformulate Formula \ref{eq:problem} as a bi-level optimization problem \cite{Xie2021Multi} as follows:
\vspace{-2mm}
\begin{align}
\label{eq:bilevel}
\min_{w,Q}\left[F(w, Q)\triangleq\frac{1}{n}\sum_{k = 1}^N n_k F_k(w(Q))\right]&, %\nonumber\\
w(Q) = \sum_{k = 1, q_k \in Q}^{N'} q_k w_k, 
\end{align}
\vspace{-4mm}

\begin{align}
\textit{s.t.} ~w_k = \argmin_{w_k}F_k(w_k)&~\forall~k \in \{1, ..., N\}, \nonumber
\end{align}
\vspace{-6mm}
\begin{align}
Q = \argmin_{Q}F(w, Q), &\sum_{q_k \in Q}^{N'}  q_k = 1, \nonumber
\end{align}
where $w$ represent the parameters for the global model, $Q$ represents the set of weights for each model. The global model is a combination of local models based on the weights $w$. $q_k$ is the weight for Model $w_k$ uploaded from Device $k$. $F$ is the loss function, which is defined in Section \ref{sec:preliminaries}. $F_k$ corresponds to the local loss function defined in Section \ref{sec:preliminaries}. We perform the outer level optimization of $Q$ to minimize $F$ by dynamically updating weights using the estimated non-IID degrees. We conduct the inner optimization of local models to minimize the local loss function. 

\subsubsection{Non-IID Degree Estimation}

While it is complicated to estimate the non-IID degree of the heterogeneous data, we propose a dynamic method to estimate the non-IID degree.
A higher non-IID degree indicates a more significant difference in the data distribution between a device and that of all the overall dataset. While the Jensen–Shannon (JS) divergence~\cite{fuglede2004jensen} can be exploited to represent the non-IID degrees, there is still a slight difference between the JS divergence and the real non-IID degree of each device. Thus, we introduce extra control parameters, \textit{i.e.}, $\upsilon_k \in \Upsilon$ and $b_k \in B$, where $\Upsilon$ and $B$ represent the set of control parameters $\upsilon_k$ and $b_k$, with the JS divergence for the non-IID degree representation as defined in Formula \ref{eq:non-iid-degree}.
\vspace{-2mm}
\begin{equation}
\vspace{-2mm}
D^{non-IID}_k(P_k, \upsilon_k, b_k) = \upsilon_k \mathcal{D}_{JS}(P_k) + b_k,
\label{eq:non-iid-degree}
\end{equation}
where $D^{non-IID}_k(P_k, \upsilon_k, b_k)$ represents the non-IID degree of Device $k$. We exploit a linear transformation of JS divergence in Formula \ref{eq:non-iid-degree} while the parameters $v_k$ and $b_k$ are dynamically updated within \TheHetName{} based on Formula \ref{eq:pUpdate} to minimize the global loss function. \ma{We conduct experiments to empirically prove the linear relationship between JS divergence and non-IID degree. See details in Apendix.} Based on Theorem \ref{eq:eqTheorem2}, a proper representation of non-IID can lead to a better model. $\mathcal{D}_{JS}(P_k)$ represents the JS divergence of Device $k$ as defined in Formula \ref{eq:js-divergence}. 
\vspace{-2mm}
\begin{equation}
\vspace{-2mm}
\mathcal{D}_{JS}(P_k) = \frac{1}{2}\mathcal{D}_{KL}(P_k||\overline{P}) + \frac{1}{2}\mathcal{D}_{KL}(P_m||\overline{P}),
\label{eq:js-divergence}
\end{equation}
where $\overline{P} = \frac{1}{2}(P_k + P_m)$, $P_m =\{P(y)|y \in \mathcal{Y}\} = \frac{\sum_{k = 1}^{N}n_k P_k}{\sum_{k = 1}^{N}n_k}$, $P_k = \{P(y)|y \in \mathcal{Y}_k\}$ with $\mathcal{Y}$ representing the label set of the total dataset, $\mathcal{Y}_k$ representing the label set of Dataset $D_k$, $P(y) = \{\theta_1, \theta_2, ..., \theta_C\}$ representing the distribution of the corresponding dataset with $\theta_i$ referring to the proportion of data samples for Class $i$ in Dataset $D_k$, and $\mathcal{D}_{KL}(\cdot||\cdot)$ is the Kullback-Leibler (KL) divergence~\cite{kullback1997information} as defined in Formula \ref{eq:kl-divergence}: 
\vspace{-2mm}
\begin{equation}
\vspace{-2mm}
\mathcal{D}_{KL}(P_i||P_j)= \sum_{y \in \mathcal{Y}}P_i(y)\log(\frac{P_i(y)}{P_j(y)}).
\label{eq:kl-divergence}
\end{equation}

\begin{figure}[t]
\vspace{-4mm}
\begin{algorithm}[H]
\caption{Dynamic Heterogeneous Federated Learning (\TheHetName{})}
\label{alg:heterAggregation}
\begin{algorithmic}[1]
\INPUT  \quad \newline
$T$: The maximum number of rounds \newline
$N$: The number of devices \newline
$P = \{P_1, P_2, ..., P_N\}$: The set of the distribution of the data on each device \newline
$\lambda = \{\lambda_1, \lambda_2, ..., \lambda_T\}$: The learning rates of control parameter update \newline
$H = \{\eta_1, \eta_2, ..., \eta_T\}$: The learning rates of model update
\OUTPUT \quad \newline
$w^{t}$: The global model at Round $t$ 
\For{$k$ in $\{1, 2, ..., N\}$ (in parallel)} \label{line:initialBegin}
\State Calculate $\mathcal{D}_{JS}(P_k)$ according to Formula \ref{eq:js-divergence}; $\upsilon_k^0 \gets 1$; $b_k^0 \gets 0$ \label{line:init}
\EndFor \label{line:initialEnd}
\For{$t$ in $\{1, 2, ..., T\}$} \label{line:serverUpdateBegin}
\State Sample clients $\mathcal{C} \subseteq \{1, 2, ..., N\}$ \label{line:deviceSampling}
\For{$k$ in $\mathcal{C}$}
\State Update $w^t_k$ according to Formula \ref{eq:localUpdate} on each device
\EndFor \label{line:localUpdateEnd}
%\State %$q_k \gets \frac{\frac{n_k}{D^{non-IID}_k(\upsilon^t_k, P_k)}}{\sum_{k \in \mathcal{C} }{\frac{n_k}{D^{non-IID}_k(\upsilon^t_k, P_k)}}}$ (Formula \ref{eq:pUpdate}) \label{line:qUpdate}
%\State %$w^t \gets \sum_{k \in \mathcal{C}} q_k w_k$ (Formula \ref{eq:bilevel}) \label{line:wUpdate}
\State Update $w^t$ according to Formula \ref{eq:bilevel} \label{line:wUpdate}
\State Update $\Upsilon^t$ and $B^t$ according to Formula \ref{eq:pUpdate} \label{line:pUpdate}
\State Update $Q$ according to Formula \ref{eq:qUpdate} \label{line:qUpdate}
\EndFor \label{line:serverUpdateEnd}
\end{algorithmic}
\end{algorithm}
\vspace{-10mm}
\end{figure}

\ma{$P_k$ values are statistical metadata of the raw data.} The statistical meta information is similar (in terms of privacy issue) to the number of samples in each device, which is transferred to the server in standard FL, e.g., FedAvg. \ma{As the statistical meta information reveals negligible privacy concerns and reveals negligible privacy concerns \cite{tan2022federated,lai2021oort,duan2019astraea}},
%lai2021oort,duan2019astraea,tan2022federated}, 
we assume that $P_k$ can be directly transferred from each device to the server in this paper.

When the statistical meta information cannot be directly transferred due to privacy issues, we can exploit the gradients from each device to estimate $P_k$, i.e., \TheHetNameE{}, which leads to similar performance as that of \TheHetName{}. Please note this is a complementary method when high security or privacy requirement is needed.
In this case, we assume that there is a small dataset with a balanced class distribution on the server, which can be obtained from public data \cite{yang2021client}. \ma{While sensitive data are not allowed to be transferred to the server, some insensitive data may reside in the servers or the cloud \cite{yang2021client}, which can be exploited as the small datasets on the server \cite{Zhang2022FedDUAP,liu2024efficient}. }
%ren2018learning,wang2021addressing,yang2021client}. 
The size of the balanced dataset is the same for each class. Please note that the balanced dataset cannot be exploited to train a satisfactory global model due to limited information. For each updated model from Device $k$, we can calculate the loss $L_i(w_k)$ with respect to Class $i$, and its corresponding gradients $\nabla L_i(w_k)$, based on the balanced dataset. We denote the set of gradients for $C$ classes as:
\vspace{-2mm}
\begin{equation}
\vspace{-2mm}
    \nabla L(w_k) = \{\nabla L_1(w_k), \nabla L_2(w_k), ..., \nabla L_C(w_k)\}.
\end{equation}
Inspired by \cite{yang2022improved}, we estimate the proportion of data samples for Class $i$ on Device $k$ as:
\vspace{-3mm}
\begin{equation}
\vspace{-2mm}
\label{eq:gradientNormalize}
    \theta_k(i) = \frac{e^{\frac{\beta}{\parallel \nabla L_i(w_k) \parallel^2}}}{\sum_{j\in C} e^{\frac{\beta}{\parallel \nabla L_j(w_k) \parallel^2}}}
\end{equation}
where $\beta$ is a hyper-parameter to adjust the proportion normalization between each class. Then, we can obtain the estimated $P_k$. This process is performed once at the beginning of the training process.

\subsubsection{Dynamic Heterogeneous Model Aggregation}

During the training process of \TheHetName{}, the models and the control parameters are updated iteratively. The local model is updated on each device using the Stochastic Gradient Descent (SGD) approach as defined in Formula \ref{eq:localUpdate} \cite{zinkevich2010parallelized},
%robbins1951stochastic, zinkevich2010parallelized}, 
where $t$ is the index of the round, $\nabla_{w^t_k} F_k(w^t_k)$ refers to the partial derivative of the model with respect to $w_t$ on Device $k$, $\eta^t$ represents the learning rate of the local training process on Device $k$. The model on each device can be a sub network generated from the global model (see details in Section \ref{subsec:fedAD}).
\vspace{-2mm}
\begin{equation}
\vspace{-2mm}
\label{eq:localUpdate}
    w^{t+1}_k \gets w^t_k - \eta^t \nabla_{w^t_k} F_k(w^t_k).
\end{equation}
The control parameters are updated on the server using Formula \ref{eq:pUpdate}, where $\lambda_\Upsilon^t$ and $\lambda_B^t$ represent the learning rates for the update of $\Upsilon^t$ and $B^t$, respectively, $w^t(Q^t)$ is defined in Formula \ref{eq:bilevel}. $\nabla_{\Upsilon^t} F(w^t(Q^t))$ and $\nabla_{b^t} F(w^t(Q^t))$ refer to the partial derivative of the global model with respect to control parameters $\Upsilon^t$ and $b^t$, respectively, which can be calculated based on the sub network of each device.
\vspace{-2mm}
\begin{equation}
\vspace{-2mm}
\begin{split}
\label{eq:pUpdate}
    \Upsilon^t &\gets \Upsilon^t - \lambda_\Upsilon^t \nabla_{\Upsilon^t} F(w^t(Q^t)), \\
    B^t &\gets B^t - \lambda_B^t \nabla_{B^t} F(w^t(Q^t)), 
\end{split}
\end{equation}
Then, the set of weights $Q^t$ is updated based on Formula \ref{eq:qUpdate}, where  $D^{non-IID}_k(P_k, \upsilon^t_k, b^t_k)$ is defined in Formula \ref{eq:non-iid-degree}. Formula \ref{eq:qUpdate} can incur higher accuracy with smaller loss compared with traditional aggregation methods (see theoretical analysis in Section \ref{subsec:theoreticalAna}). During the training process, we calculate $\nabla_{\Upsilon^t} F(w^t(Q^t))$ and $\nabla_{B^t} F(w^t(Q^t))$ only based on the selected devices ($S^t$) instead of all the devices.
\vspace{-2mm}
\begin{equation}
\vspace{-2mm}
\begin{split}
\label{eq:qUpdate}
    q^t_k = \frac{\frac{n_k}{D^{non-IID}_k(P_k, \upsilon^t_k, b^t_k)}}{\sum_{k \in S^t}{\frac{n_k}{D^{non-IID}_k(P_k, \upsilon^t_k, b^t_k)}}}, \\
    q^t_k \in Q^t,~\forall~k \in \{1, ..., N\}.
\end{split}
\end{equation}

\TheHetName{} is shown in Algorithm \ref{alg:heterAggregation}. First, the JS divergence is calculated and the control parameters are initialized (Lines \ref{line:initialBegin} - \ref{line:initialEnd}). Then, the models are updated on each selected devices (Lines \ref{line:deviceSampling} - \ref{line:localUpdateEnd}). Afterward, the global model (Line \ref{line:wUpdate}) the control parameters (Lines \ref{line:pUpdate} - \ref{line:qUpdate}) are updated. % using Formula \ref{eq:pUpdate} . 
Please refer to \cite{ghadimi2018approximation,hong2020two} 
for the convergence proof of the bi-level optimization of control parameters and the models. 

%Ji: should it be subsubsection as it belongs to Heterogeneous model aggregation?
%JL: Yes. I modified it.
\subsubsection{Theoretical Analysis}
\label{subsec:theoreticalAna}

In this section, we present the theoretical convergence proof of \TheHetName{}.
%Ji: proof of what?
%JL: proof of FedDH.
We first introduce the assumptions and then present the convergence theorems with an upper bound. 
\begin{assumption} 
\label{assump:lip}
\textit{Lipschitz Gradient: The function $F_k$ is $L$-smooth for each device $k \in \mathcal{N}$ i.e., $\parallel \nabla F_k(x) - \nabla F_k(y) \parallel  \leq  L \parallel x - y \parallel$.}
\end{assumption}
\begin{assumption} 
\label{assump:convex}
\textit{$\mu$-strongly convex: The function $F_k$ is $\mu$-strongly convex for each device $k \in \mathcal{N}$ i.e., $ <\nabla F_k(x) - \nabla F_k(y), x - y>  \geq  \mu \parallel x - y \parallel ^ 2$.}
\end{assumption}
\begin{assumption} 
\label{assump:unbiase}
\textit{Unbiased stochastic gradient: $\mathbb{E}_{\zeta^{k}_{t,h} \sim \mathcal{D}_i} [\nabla f_k(w; \zeta^{k}_{t,h})] = \nabla F_k(w) $.}
\end{assumption}
\begin{assumption} 
\label{assump:localVariance}
\textit{Bounded local variance: For each device $k \in \mathcal{N}$, the variance of its stochastic gradient is bounded: $\mathbb{E}_{\zeta^{k}_{t,h} \sim \mathcal{D}_i} \parallel \nabla f_k(w, \zeta^{k}_{t,h}) - \nabla F_k(w) \parallel^2 \le{\sigma^2}$.}
\end{assumption}
\begin{assumption} 
\label{assump:sequared}
\textit{Bounded local gradient: For each device $k \in \mathcal{N}$, the expected squared of stochastic gradient is bounded: $\mathbb{E}_{\zeta^{k}_{t,h} \sim \mathcal{D}_i}\parallel \nabla f_k(w, \zeta^{k}_{t,h}) \parallel^2 \le{G^2}$.}
\end{assumption}

%Assumption \ref{assump:lip} has been made in \cite{li2019convergence,Li2022FedHiSyn}, while Assumptions \ref{assump:unbiase}, \ref{assump:localVariance}, \ref{assump:sequared}, have been exploited in \cite{reddi2020adaptive}. 
When Assumptions \ref{assump:lip} - \ref{assump:sequared} hold, we get the following theorem.

\begin{theorem}
\label{eq:eqTheorem1}
Suppose that Assumptions \ref{assump:lip} to \ref{assump:sequared} hold and that the devices are with unbiased sampling. When the learning rate meets $\eta_{t,h} \le $min$\{\frac{1}{\mu}, \frac{1}{4L}\}$ for any $t$ and $h$, then for all $R \ge 1$ we have:
\vspace{-2mm}
\begin{equation*}
    E\parallel F(w_{T,H}) - F({w}^*) \parallel \leq \frac{L}{2}\frac{\upsilon}{TH + \gamma},
\end{equation*}
where $\gamma > 0$, $\upsilon = $max$\{(\gamma + 1)\parallel \bar{w}_{1,1} - {w}^* \parallel, \frac{\beta^{m2}B}{\beta\mu - 1}\}$, $\beta > \frac{1}{\mu}$, $B = 6L \Gamma \notag + 2Q^2 (H-1)^2 \frac{G^2}{b} + \frac{\sigma^2}{b} \notag$, $\Gamma = \max_{t\in\{1,...,T\}}\sum_{k \in S_t} p_{k,t} \Gamma_{k,t}$.
\end{theorem}

\begin{proof}
    See details in Appendix.
\end{proof}

% \vspace{-3mm}
From Theorem \ref{eq:eqTheorem1}, we can see that \TheHetName{} can converge when Assumptions \ref{assump:lip} to \ref{assump:sequared} hold. In addition, as \TheHetName{} dynamically updates the control parameters to adjust the weight of each updated model, the upper bound of the distance between the global model and the optimal model is smaller compared with other baseline methods without dynamic update based on Theorem \ref{eq:eqTheorem2}. As a result, the training efficiency of \TheHetName{} can be higher compared with the baseline methods.
\begin{theorem}
\label{eq:eqTheorem2}
Suppose that Assumptions \ref{assump:lip} to \ref{assump:sequared} hold and that the devices are with unbiased sampling. When the conditions in Theorem \ref{eq:eqTheorem1} are met and the weights $p_{k,t} = \frac{\frac{n_k}{\Gamma_{k,t}}}{\sum_{k'\in S_t}\frac{n_k'}{\Gamma_{k',t}}}$, the upper bound of $E\parallel F(w_{T,H}) - F({w}^*) \parallel$ is equal or smaller than that when $p_{k,t} = \frac{n_k}{\sum_{k'\in S_t}n_k'} $, with $(\gamma + 1)\parallel \bar{w}_{1,1} - {w}^* \parallel \leq \frac{\beta^{m2}B}{\beta\mu - 1}$.
\end{theorem}
%\vspace{-3mm}

\begin{proof}
    See details in Appendix.
\end{proof}

\begin{figure}[t!]
\vspace{-4mm}
\begin{algorithm}[H]
\caption{Federated Adaptive Dropout (\TheDropoutName{})}
\label{alg:adaptive_dropout}
\begin{algorithmic}[1]
\INPUT \quad \newline
$t$: The current round \newline
$w^*$: The current global model at Round $t$ \newline
$FL$: The list of fully connected layers in $w^*$ \newline
$CL$: The list of convolutional layers in $w^*$ \newline
$S^t$: The set of selected devices at Round $t$ \newline
$\mathcal{K}$: The set of all devices \newline
$SD^t$: The set of selected devices for the pruning rates
\OUTPUT \quad \newline
$w'^t = \{w'^t_k| k \in S^t\}$: The sub models at Round $t$
\If{should update dropout rates} \label{line:updateBegin} 
    \For{$k \in SD^t$  (in parallel)}  \label{line:localDRateBegin} 
        \State $R^t_k = \{r^{t,i}_{k,l}|l\in CL \cup FL\} \gets$ calculate the rank (weight) of each Filter (Neuron) $i$ for each Layer $l$ \label{line:rankCalculate} 
        \State Calculate the dropout rate $d_{k}$ using $w^*$ \label{line:localDRate}
    \EndFor \label{line:localDRateEnd}
    \State Estimate $ET_{max}$ using pruning rates  \label{line:estTime}
    \State Calculate aggregated $R^t$ \label{line:ranks}
    \For{$k \in \mathcal{K}$  (in parallel)}  \label{line:DRateUpdateBegin} 
        \State $d^*_{k} \gets$ calculate the proper dropout rate using $ET_{max}$ \label{line:getDRate}
        \For{$node^i \in CL$ or $FL$} \label{line:RateUpdateBegin} 
            \State $d^{i}_{k} = \frac{e^{-r^{t,i}_{k}}}{\sum_{j\in CL~\text{or}~ FL}{e^{-r^{t,j}_{k}}}} N_{CL~\text{or}~FL} d^*_k$ \label{line:nodeRateUpdate}
        \EndFor \label{line:RateUpdateEnd} 
    \EndFor \label{line:DRateUpdateEnd}
\EndIf \label{line:updateEnd} 
\For{$k \in S^t$  (in parallel)}  \label{line:dropoutBegin} 
    \For{$node^i \in CL \cup FL$}
        \State $m^{i,t}_{k} = \left\{
        \begin{aligned}
\frac{1}{1-d^{i}_{k}} & ,~w.p.~(1-d^{i}_{k}), \\
0 & ,~w.p.~d^{i}_{k}
\end{aligned}
\right.$ \label{line:dropout} 
    \EndFor
    \State $w'^t \gets$ generate the sub model using the mask variables \label{line:dropoutOP} 
\EndFor \label{line:dropoutEnd}
\end{algorithmic}
\end{algorithm}
\vspace{-10mm}
\end{figure}

\subsection{Federated Adaptive Dropout}
\label{subsec:fedAD}

In this section, we propose a neuron-adaptive dropout method, \textit{i.e.}, \TheDropoutName{}, to address the training efficiency problem. \TheDropoutName{} dynamically generates adaptive sub models, which are adapted to the various characteristics and importance of neurons, for each device while achieving excellent accuracy with multiple heterogeneous devices. 
\TheDropoutName{} simultaneously reduces the computation costs on devices and the communication costs between devices and the server. %Furthermore, \TheDropoutName{} can avoid overfitting to achieve superior performance. 
We first present the details of \TheDropoutName{}. Then, we explain the model aggregation in \TheName{}, which combines \TheHetName{} and \TheDropoutName{}.

\subsubsection{Adaptive Dropout}

\ma{As shown in Algorithm \ref{alg:adaptive_dropout}, \TheDropoutName{} is composed of two parts, \textit{i.e.}, dropout rate update (Lines \ref{line:updateBegin} - \ref{line:updateEnd}) and model dropout operation (Lines \ref{line:dropoutBegin} - \ref{line:dropoutEnd}). The dropout rate update is computationally intensive, so we only perform it when model variation indicates the need for an update. The variation of the model is calculated every fixed rounds using euclidean distance equation $\Delta^t = \sqrt{(R^t -R^{t-h})^2}$. $R^t$ is the set of the rank values and weights at Round $t$. The rank values of each filter are calculated based on the feature maps (outputs of filters) of filters \cite{lin2020hrank}, and the weight of each neuron is calculated as the sum of its corresponding input parameters in the model. The rank values and weights can well represent the importance of filters/neurons as they may have various impact on the output of the model \cite{lin2020hrank}. The dropout rate is updated once there is no decrease in variation after $\mathscr{C}$ rounds to avoid accuracy degradation in critical learning periods \cite{yan2022seizing}. Additionally, the update is performed again when a reduction in variation is observed at a specific round.}

\ma{Within the dropout rate update process, the rank value or the weight (Line \ref{line:rankCalculate}) of each filter or neuron and the dropout rate  (Line \ref{line:localDRate}) of the global model is calculated cross selected devices $SD$.
%, \textit{i.e.}, the one with the most significant dataset, the one with the smallest dataset, and a randomly selected one. %The rank of a filter is calculated based on the feature maps of filters in the $l$-th layer \cite{lin2020hrank}, and the weight of a neuron is calculated based on the weights within the global model. 
For each Device $k$, after $t$ rounds of training, the parameters are updated to $W'_{k}$ from the initial parameters $W_{k}$ and the difference is denoted $\Delta_{k} = W_{k} - W'_{k}$. To capture the behavior of the local loss function, we calculate the Hessian matrix of the local loss function, i.e., $H(W'_{k})$, with the eigenvalues sorted in ascending order, \textit{i.e.}, $\{\mathscr{E}^m_{k}|m \in (1, h_{k})\}$, where $h_{k}$ represents the rank of the Hessian matrix and $m$ refers to the eigenvalue index. Next, We construct a function $B_{k}(\Delta_{k}) = H(W'_{k}) - \triangledown L(W_{k})$, where $\triangledown L(\cdot)$ is the gradient of the local loss function. We denote the Lipschitz constant of $B_{k}(\Delta_{k})$ by $\mathscr{L}_{k}$. %Inspired by \cite{zhang2021validating}, 
To avoid possible accuracy degradation \cite{zhang2021validating}, we select the first eigenvalue $m_{k}$ that satisfies the condition $\mathscr{E}_{m_{k+1}} - \mathscr{E}_{m_{k}} > 4\mathscr{L}_{k}$. Then, we calculate the dropout rate $d_{k} = \frac{m_{k}}{h_{k}}$. Afterward, we can calculate the aggregated dropout rate for the devices by $d_{avg} = \sum_{k\in SD}{\frac{n_{k}}{\sum_{k\in SD}n_k}d_k}$, which serves as the dropout rate on the device with the most significant dataset. 
%Then, for each device, we can estimate the execution time of a round by $ET_k(p_k) = ET^o_k * p_k$, where $ET^o_k$ is the execution time based on previous profiling information, and 
We can get the execution time $ET_{max} = ET^o_{max} * d_{avg}$ of this device (Line \ref{line:estTime}), with $ET^o_{max}$ representing the execution without dropout based on previous profiling information. Furthermore, we aggregate the rank values using $R^t = \sum_{k\in SD}{\frac{n_{k}}{\sum_{k\in SD}n_k}R^t_k}$ (Line \ref{line:ranks}). With this, the appropriate dropout rate for each device is recalculated as $d^*_k = $min$\{\frac{ET_{max}}{ET^o_k}, 1\}$ (Line \ref{line:getDRate}). The dropout rate of each filter or neuron $d^i_k$ is then updated based on its rank values or weights $r^{t,i}_{k} \in R^t$ (Lines \ref{line:RateUpdateBegin} - \ref{line:RateUpdateEnd}). In Line \ref{line:nodeRateUpdate}, $N_{CL~or~FL}$ represents the number of filters or neurons in 
all convolutional layers $CL$ or fully connected layers $FL$. Please note that the update operation (Lines \ref{line:RateUpdateBegin} - \ref{line:RateUpdateEnd}) is carried out separately for convolutional layers and fully connected layers.}

Finally, we perform the dropout operation with the updated dropout rates (Lines \ref{line:dropoutBegin} - \ref{line:dropoutEnd}). The mask variables $m^{i,t}_{k}$ are calculated based on the probability $d^i_{k}$ (Line \ref{line:dropout}). Within the dropout operation, when $m^{i,t}_{k}=0$, the filter or the neuron is removed, and otherwise, the weight is updated (Line \ref{line:dropoutOP}). \ma{In addition, while \TheDropoutName{} is primarily discussed in the context of traditional neural network architectures, it is designed to be adaptable and extendable to a variety of model types, including Transformers. The key principles of \TheDropoutName{}-adaptation of dropout rates base on model characteristics-are not inherently limited to convolutional layers or fully connected networks. In the case of Transformers, \TheDropoutName{} can be applied to dense layers and self/cross multi-head attention layers. Specifically, the adaptive dropout mechanism can selectively deactivate neurons in dense layers or individual attention heads based on their computed importance scores. By targeting these components, the method maintains computational efficiency while preserving the overall integrity and performance of the model.}

\ma{\TheDropoutName{} updates the dropout rates of each filter or neuron based on its rank values or weights, determining which filters or neurons in each layer should be dropped out with certain possibility. Within the dropout operation, when $m^{i,t}_{k}=0$, the filter or the neuron is removed, and otherwise, the weight is updated (Line \ref{line:dropoutOP}). The dropout operation of a layer $l$ is a special case, i.e., $m^{i,t}_{k}=0$ for $i \in Layer~l$. Although this may be efficient in calculating the dropout rate, it may bring significant accuracy degradation due to the absence of a whole layer. }

\subsubsection{Model Aggregation}

In order to aggregate the heterogeneous models, we take an efficient model aggregation method in \TheName{}. For the neurons or filters that are not removed (dropped) in all the selected devices, we directly exploit Formula \ref{eq:bilevel} to calculate the aggregated value as follows: $m^i = \sum q_k * m^i_k$, where $m^i$ represents the parameter in the global model, $q^i_k$ represents the weight of Model $k$ in $Q$ and $m^i_K$ represents the parameter in local Model $k$. If the neurons or filters are removed (dropped) in Set Drop, we calculate the aggregated weight based on the following formula: $m^i = (\sum_{k \notin Drop}q_k * m^i_k) / (\sum_{k \notin Drop} q_k)$, which is the sum of the valid weighted parameters (the corresponding filter or neuron is not dropped on the devices) divided by the sum of weights corresponding to the devices of valid parameters (the corresponding filter or neuron is not dropped on the devices).

\ma{
\subsection{The Integration of \TheHetName{} and \TheDropoutName{}}

In this section, we detail the integration of the proposed methods \TheHetName{} and \TheDropoutName{} into \TheName{} framework. The \TheName{} framework synergistically combines the strengths of \TheHetName{} and \TheDropoutName{} to addresses both the model performance and training efficiency problems inherent in FL with non-IID data distributions and heterogeneous devices.

Specifically, during each training round in \TheName{}, the server begins by performing the neuron-adaptive dropout mechanism of \TheDropoutName{} to generate sub-models for the selected devices. The dropout rates for each filter or neuron are dynamically adjusted based on the importance of filters or neurons and the resource constraints of the device. The generated sub-models are then distributed to the selected devices, where local training is conducted using their respective local datasets. After local training, the devices upload their updated sub-models back to the server. The server then performs dynamic heterogeneous model aggregation using the \TheHetName{} method, which exploits the non-IID degrees of the heterogeneous data on the selected devices to adjust the aggregation weights. The integration offers synergistic benefits. The slight computational overhead introduced by \TheHetName{} is effectively compensated by the efficiency gains from \TheDropoutName{}, resulting in accelerating the training process and improving model accuracy.

\m{The integration of \TheHetName{} and \TheDropoutName{} into the \TheName{} framework establishes a synergistic relationship that simultaneously addresses statistical data heterogeneity and device heterogeneity in FL. While both methods demonstrate effectiveness independently, their combination yields complementary advantages. \TheHetName{} addresses model performance challenges through dynamic weight adjustment based on non-IID degrees, while \TheDropoutName{} optimizes communication and computation efficiency via adaptive dropout operations. The computational overhead introduced by the dynamic weight calculations of \TheHetName{} is compensated by the efficiency gains from the reduced model size of \TheDropoutName{}. Furthermore, the weighted aggregation mechanism of \TheHetName{} mitigates potential accuracy degradation from the dropout operations of \TheDropoutName{} by appropriately valuing the contribution of each device according to its data distribution characteristics. Additionally, the resource optimization of \TheDropoutName{} enables more training rounds within constrained time budgets, providing \TheHetName{} with increased opportunities to refine weight adjustments. Thus, the combination of these two approaches in \TheName{} yields superior performance in terms of both accuracy and efficiency compared to either method individually or state-of-the-art baseline approaches.}
}
%\begin{table*}
%  \caption{The accuracy and training time with \TheDropoutName{} and various baseline methods. ``Accuracy'' represents the accuracy of the final global model. ``Time'' represents the training time (s) to achieve the accuracy of 0.44 for LeNet and 0.723 for CNN. ``MFLOPs'' represents the computation costs. ``NaN'' represents that the accuracy does not achieve the target accuracy.
%  }
%  \label{tab:cmp_dropout}
%  \centering
%  \begin{tabular}{l|lll|lll}
%    \toprule
%    \multirow{2}{*}{Method} & \multicolumn{3}{c|}{LeNet}  & \multicolumn{3}{c}{CNN}   \\
%    \cmidrule(r){2-7}
%                            & Accuracy  & Time  & MFLOPs  & Accuracy  & Time & MFLOPs \\
%    \midrule
%    \TheDropoutName{}       &  \textbf{0.464} & \textbf{1161}  & \textbf{4.81}   &  \textbf{0.732}    & \textbf{3825}   & 4.36 \\
%    FedAvg                  &  0.457    & 2166  & 6.58   &  0.730    & 5301  & 4.60 \\
%    AFD                     &  0.377    & NaN   & 4.94   &  0.589    & NaN  & \textbf{3.45} \\
%    FedDrop                 &  0.444    & 3261  & 6.15   &  0.725    & 5438  & 4.46 \\
%    \bottomrule
%  \end{tabular}
%\end{table*}

\begin{table*}[t]
\tiny
\vspace{-4mm}
  % \caption{The accuracy and training time with \TheName{}, \TheHetName{}, \TheDropoutName{}, and various baseline methods on CIFAR-10 and CIFAR-100. ``Acc'' represents the accuracy of the final global model. ``Time'' represents the training time (s) to achieve the accuracy of 0.54 for LeNet/CIFAR-10, 0.72 for CNN/CIFAR-10, 0.26 for LeNet/CIFAR-100, and 0.41 for CNN/CIFAR-100. ``MFLPs'' represents the computation costs (``MFLOPs''). ``NaN'' represents that the accuracy does not achieve the target accuracy. \TheName{}, \TheHetName{}, and \TheDropoutName{} are our proposed methods. The best results are highlighted in \textbf{bold} and the second best results are highlighted with \underline{underline} and \emph{italic}.
  % }
  \caption{Performance comparison on CIFAR-10 and CIFAR-100 datasets. \TheName{}, \TheHetName{}, and \TheDropoutName{} are our proposed methods. The best results are highlighted in \textbf{bold} and the second best results are highlighted with \underline{underline} and \emph{italic}.
}
  % \vspace{-4mm}
  \label{tab:cmp_DHAD_10}
  \centering
  \begin{tabular}{c|ccc|ccc|ccc|ccc}
    \toprule
    \multirow{3}{*}{Method} & \multicolumn{6}{c|}{CIFAR-10}  & \multicolumn{6}{c}{CIFAR-100} \\
    \cmidrule(r){2-13} & \multicolumn{3}{c|}{LeNet}  & \multicolumn{3}{c|}{CNN} & \multicolumn{3}{c|}{LeNet}  & \multicolumn{3}{c}{CNN} \\
    \cmidrule(r){2-13} & Acc  & Time  & MLPs  & Acc  & Time & MLPs & Acc  & Time  & MLPs  & Acc  & Time & MLPs \\
    \midrule
    \textbf{\TheName{}}      (ours)   &  \textbf{0.633}    & \textbf{1371}   & 0.618   &  \textbf{0.749}    & \textbf{3074}   & 4.324  &  \textbf{0.300}    & \textbf{2127}   & 0.609   &  \textbf{0.423}    & \textbf{4952}   & 4.416\\
    \textbf{\TheHetName{}}   (ours)     &  \underline{\textit{0.621}}    & 1426   & 0.652    &  \underline{\textit{0.744}}    & \underline{\textit{3322}}   & 4.549    &  \underline{\textit{0.297}}    & 2495   & 0.659    &  \underline{\textit{0.420}}    & \underline{\textit{5018}}   & 4.554\\
    \textbf{\TheDropoutName{}}  (ours)  &  0.595    & \underline{\textit{1394}}   & 0.615   &  0.732    & 3632   & 4.324   &  0.295    & \underline{\textit{2175}}   & 0.604   &  0.412    & 5324   & 4.330 \\
    FedAvg                 &  0.565    & 2366  & 0.652   &  0.730    & 5265  & 4.549      &  0.263    & 5714  & 0.659   &  0.414    & 5089  & 4.554\\
    FedProx                &  0.588    & 1994   & 0.652   &  0.730    & 3912  & 4.549    &  0.274    & 3750   & 0.659   &  0.415    & 5613  & 4.554 \\
    FedNova                &  0.582    & 2200  & 0.652   &  0.728    & 5296  & 4.549    &  0.286    & 3082  & 0.659   &  0.416    & 5514  & 4.554 \\
    MOON                   &  0.610    & 2871   & 0.689   &  0.741    & 4513  & 4.584     &  0.281    & 3682   & 0.712   &  0.405    & 7031  & 4.628\\
    FedDyn                 &  0.543    & 2256   & 0.652   &  0.700    & NaN  & 4.549    &  0.228    & NaN   & 0.659   &  0.344    & NaN  & 4.554\\
    FedDST                 &  0.582    & 2178   & 0.645   &  0.645    & NaN  & 4.504    &  0.168    & NaN   & 0.652   &  0.181    & NaN  & 4.508\\
    PruneFL                &  0.392    & NaN   & 0.652   &  0.524    & NaN  & 4.549    &  0.127    & NaN   & 0.659   &  0.187    & NaN  & 4.554\\
    % FedDUAP                &  0.521    & NaN   & \textbf{0.411}   &  0.661    & NaN  & \textbf{2.866}    &  0.262    & 4090   & \textbf{0.415}   &  0.357    & NaN  & \textbf{2.874}\\
    AFD                    &  0.511    & 3909   & 0.637   &  0.589    & NaN   & 4.540   &  0.270    & 3768   & 0.644   &  0.228    & NaN   & 4.550\\
    FedDrop                &  0.585    & 3993  & 0.636   &  0.725    & 5385  & 4.533     &  0.254    & NaN  & 0.640   &  0.405    & 6540  & 4.538 \\
    FjORD                  &  0.593   & 1663     & \textbf{0.482}       & 0.723    & 4556  & \textbf{3.474}     &  0.292    & 2262  & \textbf{0.489}  &  0.369    & NaN  & \textbf{3.479}\\
    MOON+FjORD                &  0.608    & 2821  & \underline{\textit{0.526}}   &  0.737    & 4609  & \underline{\textit{3.959}} & 0.288 & 3376 & \underline{\textit{0.518}} &  0.416 & 5424 & \underline{\textit{4.043}}    \\
    \bottomrule
  \end{tabular}
  % \vspace{-8mm}
\end{table*}

\section{Experimental Evaluation}
\label{sec:experiments}

\m{In this section, we compare \TheName{} with 13 state-of-the-art baseline approaches, exploiting four models and four datasets. We choose baseline methods that specifically address data and device heterogeneity challenges in FL. We includes methods (e.g. FedAvg \cite{mcmahan2017communication}, FedProx \cite{Li2020FedProx}, and FedNova \cite{Wang2020Tackling}) which set foundational standards for handling non-IID data, alongside recent advancements (e.g. MOON \cite{li2021model}, FedDyn \cite{acar2021federated}, FedAS \cite{yang2024fedas}, FedKTL \cite{zhang2024upload} and AugFL \cite{yue2025augfl} ). We selected methods (e.g. FedDST \cite{bibikar2022federated}, PruneFL \cite{jiang2022model}, AFD \cite{Bouacida2021FedDropout}, FedDrop \cite{wen2022federated}, and FjORD \cite{horvath2021fjord}) for device heterogeneity. Additionally, we integrated combinations (e.g. MOON and FjORD), ensuring a comprehensive and fair comparison.} First, we present the experimental setup. Then, we present our experimentation results.

\subsection{Experimental Setup}
\label{subsec:expSetup}

We consider a standard FL environment composed of a server and $100$ devices, each with its  stored  data, and we randomly select $10$ devices in each round. We simulate this environment
using 44 Tesla V100 GPU cards. 
We use the datasets of CIFAR-10\cite{krizhevsky2009learning}, CIFAR-100 \cite{krizhevsky2009learning}, SVHN \cite{netzer2011reading}, and TinyImageNet \cite{le2015tiny}. We report the results of four models, i.e., a simple synthetic CNN network (CNN), LeNet5 (LeNet) \cite{lecun1989handwritten}, VGG11 (VGG) \cite{simonyan2015very}, and ResNet18 (ResNet) \cite{he2016deep}.
The complexity of the four models are: ResNet (11,276,232 parameters) $>$ VGG (9,750,922 parameters) $>$ CNN (122,570 parameters) $>$ LeNet (62,006 parameters).
%, and VGG11 (VGG) \cite{simonyan2014very}. 
We utilize a decay rate in both the training process of devices and the control parameter adjustment. We take 500 as the maximum number of rounds for LeNet, CNN, and VGG and 1000 as that for ResNet. We fine-tune the hyper-parameters and present the best performance for each approach. %In addition, we carry out the experimentation three times and report the averaged accuracy. 
The accuracy, the training time to achieve target accuracy, and the computation costs are shown in Tables \ref{tab:cmp_DHAD_10} and \ref{tab:cmp_DHAD_100}. \ma{We report the training time (s) to achieve the specific accuracies. These target accuracies were selected based on empirical observations of the training process and represent points where the models demonstrate significant learning progress while still having room for further improvement. Instead of using a predefined standard, such as the average accuracy of all methods or a fraction of the accuracy of the best method, these specific values were chosen to provide a clear and practical benchmark for comparing the convergence speeds of different methods.}

We exploit four datasets including CIFAR-10 (60,000 images with 10 classes), CIFAR-100 (60,000 images with 100 classes), SVHN (99289 images with 10 classes), and TinyImageNet (100,000 images with 200 classes). We allocate each device a proportion of the samples of each label according to Dirichlet distribution \cite{LiStudy2022}. Specifically, we sample $p_k \sim Dir(\beta)$ and allocate a proportion of the samples in class $k$ to Device $j$. $Dir(\beta)$ is the Dirichlet distribution with a concentration parameter $\beta$ (0.5 by default), the values of which corresponds to various non-IID degrees of the data in each device. We exploit the same method to handle the other datasets. We use a learning rate of 0.1 for the model update, a learning rate decay of 0.99 for the model update, a batch size of 10, and a local epoch ($\tau$) of 5. The values of the other hyper-parameters are shown in Table \ref{tab:parameters}. 
%We sort the data according to the label and then divide these data evenly into $200$ fractions. Each device is assigned various numbers of fractions, ranging from 1 - 10, with various sizes. We randomly choose some devices to share several amount of the same fractions in order to improve the diversity of non-IID degrees among multiple devices for the experimentation. 

\begin{table*}[t!]
\tiny
\vspace{1mm}
%   \caption{The accuracy and training time with \TheName{}, \TheHetName{}, \TheDropoutName{}, and various baseline methods on SVHN and TinyImageNet. ``Acc'' represents the accuracy of the final global model. ``Time'' represents the training time (s) to achieve the accuracy of 0.85 for LeNet, 0.88 for CNN, 0.91 for VGG, and 0.36 for ResNet. ``Time$^*$'' represents $*10^2$s. ``MLPs'' represents the computation costs (``MFLOPs''). ``NaN'' represents that the accuracy does not achieve the target accuracy. \TheName{}, \TheHetName{}, and \TheDropoutName{} are our proposed methods. The best results are highlighted in \textbf{bold} and the second best results are highlighted with \underline{underline} and \emph{italic}.
% }
    \caption{Performance comparison on SVHN and TinyImageNet datasets. ``Time$^*$'' represents $*10^2$s. \TheName{}, \TheHetName{}, and \TheDropoutName{} are our proposed methods. The best results are highlighted in \textbf{bold} and the second best results are highlighted with \underline{underline} and \emph{italic}.
}
  \label{tab:cmp_DHAD_100}
  % \vspace{-4mm}
  \centering
  \begin{tabular}{c|ccc|ccc|ccc|ccc}
    \toprule
    \multirow{3}{*}{Method} & \multicolumn{9}{c|}{SVHN}  & \multicolumn{3}{c}{TinyImageNet} \\
    \cmidrule(r){2-13} & \multicolumn{3}{c|}{LeNet}  & \multicolumn{3}{c|}{CNN} & \multicolumn{3}{c|}{VGG} & \multicolumn{3}{c}{ResNet}\\
    \cmidrule(r){2-13} & Acc  & Time  & MLPs  & Acc  & Time & MLPs & Acc  & Time$^*$ & MLPs & Acc  & Time$^*$ & MLPs \\
    \midrule
    \textbf{\TheName{}} (ours)     &  \textbf{0.8760}    & \textbf{1645}   & 0.620   &  \textbf{0.9110}    & \textbf{1177}   & \underline{\textit{4.127}} & \textbf{0.9367} & \textbf{131} & 138.9 & \textbf{0.3765}  & 859 & 474.3\\
    \textbf{\TheHetName{}} (ours)  &  \underline{\textit{0.8755}}    & \underline{\textit{1777}}   & 0.652    &  \underline{\textit{0.9093}}    & \underline{\textit{1228}}   & 4.549 & \underline{\textit{0.9321}} & \underline{\textit{132}} & 153.3 & 0.3724 & 1003 & 558.0\\
    \textbf{\TheDropoutName{}} (ours)  &  0.8665    & 2246   & 0.615   &  0.9056    & 1650   & 4.172 & 0.9309 & 141 & 139.4 & \underline{\textit{{0.3763}}} & \textbf{615} & \underline{\textit{472.3}}\\
    FedAvg              &  0.8663    & 2347  & 0.652   &  0.9055    & 1595  & 4.549  & 0.9286 & 154 & 153.3 & 0.3638 & 1276 & 558.0\\
    FedProx             &  0.8652    & 2374   & 0.652   &  0.8999    & 1864  & 4.549 & 0.9260 & 159 & 153.3 & 0.3694 & 1256 & 558.0\\
    FedNova             &  0.8663    & 2333  & 0.652   &  0.9023    & 1454  & 4.549 & 0.9287 & 149 & 153.3 & 0.3598 & 1579 & 558.0\\
    MOON                &  0.8550    & 4843   & 0.689   &  0.8849    & 3558  & 4.584 & 0.1870 & NaN & 153.2 & 0.1536 & NaN & 558.0\\
    FedDyn              &  0.8338    & NaN   & 0.652   &  0.8813    & 2174  & 4.549 & 0.1646 & NaN & 153.3 & 0.0792 & NaN & 558.0\\
    FedDST              &  0.8554    & 2774   & 0.647   &  0.8845    & 1906  & 4.299 & 0.1851 & NaN & 148.7 & 0.3014 & NaN & \textbf{446.4}\\
    PruneFL              &  0.8450    & NaN   & 0.652   &  0.8994    & 1668  & 4.549 & 0.1646 & NaN & 153.3 & 0.3322 & NaN & 558.0\\
    % FedDUAP              &  0.8040    & NaN   & \textbf{0.411}   &  0.8210    & NaN  & \textbf{2.735} & 0.8181 & NaN & \textbf{106.5} & 0.0792 & NaN & \textbf{250.8}\\
    AFD                 &  0.8504    & 3185   & 0.648   &  0.5046    & NaN   & 4.531 & 0.7683 & NaN & 152.9 & 0.3727 & \underline{\textit{761}} & 557.7\\
    FedDrop             &  0.8568    & 3031  & 0.636   &  0.8985    & 1728  & 4.533 & 0.9285 & 170 & 152.3 & 0.3748 & 865 & 557.2\\
    FjORD             &  0.8434    & 2185  & \textbf{0.482}   &  0.900    & 1730  & \textbf{3.959} & 0.9107 & 154 & \textbf{116.5} & 0.3463 & NaN & 489.7 \\
    MOON+FjORD             &  0.8414    & NaN  & \underline{\textit{0.528}}   &  0.8964    & 1765  & 4.279 & 0.1992 & NaN & \underline{\textit{118.4}}& 0.1456 & NaN & 536.4 \\
    \bottomrule
  \end{tabular}
  % \vspace{-4mm}
\end{table*}

The accuracy and training time with \TheHetName{}, \TheDropoutName{} and various baseline methods for LeNet and CNN with CIFAR-10 and CIFAR-100 are shown in Figure \ref{fig:cmp_dhad_cifar}, and the results for LeNet, CNN, and VGG with SVHN and ResNet with TinyImageNet are shown in Figure \ref{fig:cmp_dhad_svhn_tinyimagenet}. We adaptively attribute dropout rates for each filter and neuron, which yields stable performance. The results of accuracy and rounds are shown in Figures \ref{fig:acc_rounds_dhad}, which reveal \TheName{} corresponds to more stable training process than other methods. 

\m{We evaluate the performance of our proposed methods using three key metrics: accuracy, training time, and computation cost. Accuracy measures the proportion of correctly classified samples in the test dataset using the trained global model. Training time is measured in seconds (s) and represents the time required to achieve a target accuracy level, which varies by model and dataset combination. For CIFAR-10, we set target accuracies of 0.54 for LeNet and 0.72 for CNN; for CIFAR-100, 0.26 for LeNet and 0.41 for CNN; for SVHN, 0.85 for LeNet, 0.88 for CNN, and 0.91 for VGG; and for TinyImageNet, 0.36 for ResNet. Computation cost is measured in Million Floating Point Operations (MFLOPs), which quantifies the computational complexity of the model training process. This metric is calculated as the sum of all floating-point operations required for forward and backward passes during model training, divided by $10^6$. Lower MFLOPs values indicate more computationally efficient methods.}

\subsection{Evaluation of Our Approach}

In this section, we present the experimental results for \TheHetName{} and \TheDropoutName{}. Then, we show the evaluation results of the combination of \TheHetName{} and \TheDropoutName{}, i.e., \TheName{}.

\subsubsection{Evaluation of \TheHetName{}} In this section, we compare the accuracy of \TheHetName{} with FedAvg \cite{mcmahan2017communication}, FedProx \cite{Li2020FedProx}, FedNova \cite{Wang2020Tackling}, MOON \cite{li2021model}, FedDyn \cite{acar2021federated}, \m{FedAS \cite{yang2024fedas}, FedKTL \cite{zhang2024upload} and AugFL \cite{yue2025augfl}}. \TheHetName{} dynamically adjusts the weights of uploaded models with the non-IID degree of the heterogeneous data in the model aggregation process, corresponding to superior performance.
As shown in Table \ref{tab:cmp_DHAD_10}, \TheHetName{} achieves a significantly higher accuracy compared with FedAvg (5.6\%), FedProx (3.3\%), FedNova (4.0\%), MOON (1.1\%), and FedDyn (7.8\%) for LeNet with CIFAR-10. We find that FedDyn leads to lower accuracy than FedAvg. Although FedNova, FedProx, and MOON outperform FedAvg, they correspond to lower accuracy than \TheHetName{} due to simple model aggregation. In addition, we find similar results for CNN with CIFAR-10, while the advantages of \TheHetName{} (from 0.3\% to 4.5\%) are relatively smaller than those of LeNet. Furthermore, Table \ref{tab:cmp_DHAD_10} reveals significant advantages (up to 3.4\% for FedAvg, 2.3\% for FedProx, 1.1\% for FedNova, 2.2\% for MOON, and 8.3\% for FedDyn) of \TheHetName{} for LeNet and CNN with CIFAR-100. As shown in Table \ref{tab:cmp_DHAD_100}, we can observe that \TheHetName{} corresponds to much higher accuracy compared with FedAvg (up to 0.9\%), FedProx (up to 1.0\%), FedNova (up to 1.3\%), MOON (up to 74.5\%), and FedDyn (up to 76.8\%) while dealing with SVHN and TinyImageNet. %Furthermore, we find that the difference between \TheHetName{} and \TheHetNameE{} is less than 0.8\% with LeNet and CNN on CIFAR-10 and CIFAR-100 as shown in Figures \ref{fig:cmp_dhe_lenet_10} and \ref{fig:cmp_dhe_cnn_10}, which reveals that we can avoid transferring meta information with negligible accuracy degradation.

\subsubsection{Evaluation of \TheDropoutName{}}
\label{subsec:dropout}

In this section, we present the comparison results between \TheDropoutName{} and FedAvg \cite{mcmahan2017communication}, Federated Dynamic Sparse Training (FedDST) \cite{bibikar2022federated}, PruneFL \cite{jiang2022model},  Adaptive Federated Dropout (AFD) \cite{Bouacida2021FedDropout}, Federated Dropout (FedDrop) \cite{wen2022federated}, and FjORD \cite{horvath2021fjord}, in terms of both the accuracy and efficiency. We take the dropout rates of 0.25 for AFD as reported in \cite{Bouacida2021FedDropout} and 0.5 for FedDrop, which corresponds to the highest accuracy reported in \cite{wen2022federated}. In addition, we fine-tune the hyper-parameters of FjORD to achieve the best accuracy.
%, in our experimentation. 

We find that the total training time of \TheDropoutName{} is much shorter than that of FedAvg (up to 8.6\%), and FedDrop (up to 4.5\%), AFD (up to 21.9\%), but longer than that of FjORD (up to 27.6\%).
However, as shown in Table \ref{tab:cmp_DHAD_10}, the accuracy of \TheDropoutName{} is significantly higher than that of FedAvg (up to 3.2\%), FedDrop (up to 4.1\%), AFD (up to 18.4\%), and FjORD (up to 4.3\%) for LeNet and CNN with CIFAR-10 and CIFAR-100. In addition, \TheDropoutName{} leads to the shortest training time (up to 31.0\% compared with FedAvg, up to 36.0\% compared with FedDST, 55.1\% compared with FedDrop, and 25.8\% compared with FjORD) to achieve the target accuracy with LeNet (AFD and PruneFL cannot achieve the target accuracy) and corresponds to smaller computation costs compared with FedAvg (up to 8.6\%), FedDST (up to 7.4\%), PruneFL (up to 9.1\%), AFD (up to 6.6\%), and FedDrop (up to 4.5\%) because of the adaptive dropout rates. We can get similar results on LeNet, CNN, and VGG over SVHN and ResNet over TinyImageNet as shown in Table \ref{tab:cmp_DHAD_100}. \ma{\TheDropoutName{} is designed to reduce the communication and computation costs in order to accelerate the training process while achieving excellent accuracy with heterogeneous devices. As shown in Table \ref{tab:cmp_DHAD_10} and Table \ref{tab:cmp_DHAD_100}, the total training time of \TheDropoutName{} is much shorter than that of \TheHetName{} (up to 38.7\%). While the total training time of \TheHetName{} is shorter sometimes, \TheDropoutName{} leads to consistent lower computation costs (up to 15.4\%) compared with \TheHetName{}.}

\begin{table*}[t]
\tiny
\vspace{-3mm}
\caption{Values of hyper-parameters in the experimentation. ``TI'' represents TinyImageNet.}
\vspace{-5mm}
\label{tab:parameters}
\begin{center}
\begin{tabular}{c|c|c|c|c|c|c|c|c|c}
\toprule
\multicolumn{2}{c|}{\multirow{3}{*}{Name}} & \multicolumn{8}{c}{Values} \\
\cline{3-10}
\multicolumn{2}{c|}{}& \multicolumn{3}{c|}{LeNet} & \multicolumn{3}{c|}{CNN}  & \multicolumn{1}{c|}{VGG}  & \multicolumn{1}{c}{ResNet}\\
\cline{3-10}
\multicolumn{2}{c|}{}& CIFAR-10 & CIFAR-100 & SVHN                  & CIFAR-10   & CIFAR-100  & SVHN                  & SVHN & \multicolumn{1}{c}{TI } \\
\hline
%\multicolumn{2}{c|}{$\mathscr{C}$} \multicolumn{2}{c|}{40} & 0 &180&\multicolumn{2}{c|}{0} & 0& \multicolumn{1}{c}{40} \\
%\hline
\multirow{4}{*}{FedDH \& FedAD} & $\lambda_{\gamma}$ & 0.0001 & 0.01 & 0.01 & 0.001 & 0.001 & 0.01 & 0.001 & 0.1\\
& $decay_{\gamma}$ & 0.999 & 0.99 & 0.999 & 0.9999 & 0.9999& 0.999 & 0.9 & 0.9\\
& $\lambda_{B}$ & 0.0001 & 0.1 & 0.01 & 0.1 &0.0001 & 0.001 & 0.0001 & 0.1\\
& $decay_B$ & 0.99 & 0.99 & 0.99 & 0.99 & 0.99 & 0.99 & 0.999 & 0.999\\
\hline
\multirow{4}{*}{FedDHAD} & $\lambda_{\gamma}$ & 0.01 & 0.01 & 0.0001 & 0.0001 & 0.01 & 0.1 & 0.1 &0.1\\
& $decay_{\gamma}$ & 0.9999 & 0.9999 & 0.999 & 0.999 &0.9999 & 0.99 & 0.99 & 0.9999\\
& $\lambda_{B}$ & 0.0001 & 0.0001 & 0.1 & 0.1  & 0.0001 & 0.01 & 0.01 & 0.001\\
& $decay_B$ & 0.99 & 0.99 & 0.99 & 0.99 & 0.99 & 0.99 & 0.99 & 0.99\\
\bottomrule
\end{tabular}
\vspace{-6mm}
\end{center}
\end{table*}

While the computation costs of FjORD are smaller compared with \TheDropoutName{}, \TheDropoutName{} leads to significantly higher accuracy (up to 4.3\%).
As FedDrop is only applied to the fully connected layers, it corresponds to a longer training time compared with AFD, FjORD, and \TheDropoutName{}. Since FedDrop only considers random dropout operation, its accuracy is slightly lower than that of FedAvg for both LeNet (up to 0.9\%) and CNN (up to 0.5\%). While AFD considers the influence of dropout on the loss, it ignores the phenomenon that the loss is naturally reduced during the training process and thus corresponds to significantly inferior accuracy (up to 18.4\% compared with \TheDropoutName{}, 18.6\% compared with FedAvg, and 17.7\% compared with FedDrop). Although FjORD can adaptively tailor the model width for heterogeneous devices, it cannot achieve lossless pruning and corresponds to inferior accuracy (up to 4.3\%) compared with \TheDropoutName{}. \TheDropoutName{} adaptively calculates the dropout rate,
%for each neuron or filter, 
which leads to excellent final accuracy within a short training time. \ma{Notably, the use of VGG model in SVHN results in significantly lower accuracy, primarily due to the inherent complexity and the large number of parameters (e.g. 9,750,922 parameters) of VGG. This complexity makes VGG vulnerable to the challenges of data heterogeneity. Furthermore, PruneFL employs pruning methods that involve lossy operations, which can further degrade accuracy.}

\subsubsection{Evaluation of \TheName{}}

In this section, we integrate \TheHetName{} and \TheDropoutName{} into \TheName{}, which we
evaluate in comparison with FedAvg, FedProx, FedNova, MOON, FedDyn, \m{FedAS, FedKTL, AugFL,} FedDST, PruneFL, AFD, FedDrop, FjORD and a combination of MOON and FjORD, in terms of both accuracy and efficiency. 

\begin{figure*}[t]
\centering
\subfigure[LeNet \& CIFAR-10]{
\includegraphics[width=0.22\linewidth]{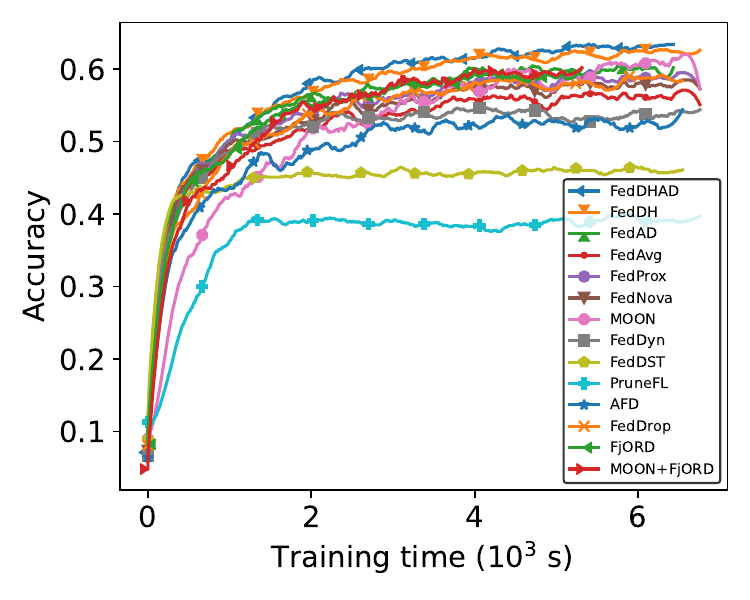}
\label{fig:cmp_dhad_lenet_10}
}
\vspace{-1mm}
\subfigure[CNN \& CIFAR-10]{
\includegraphics[width=0.22\linewidth]{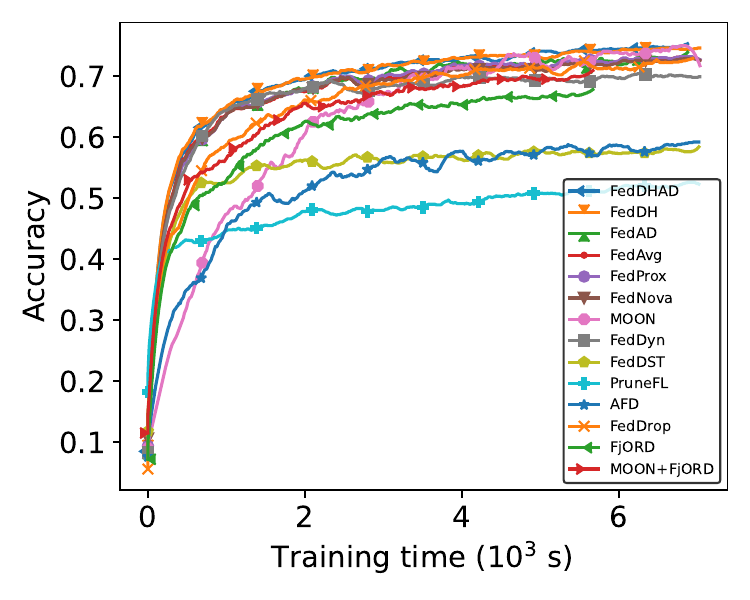}
\label{fig:cmp_dhad_cnn_10}
}
\vspace{-1mm}
\subfigure[LeNet \& CIFAR-100]{
\includegraphics[width=0.22\linewidth]{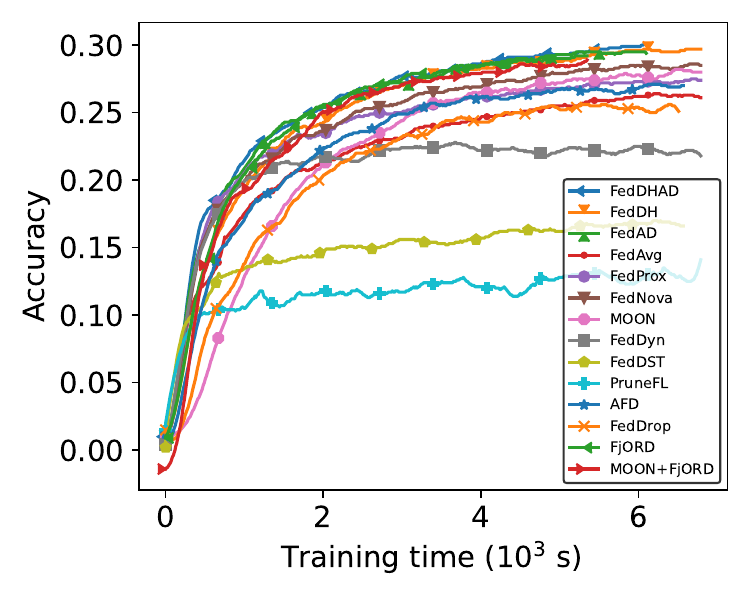}
\label{fig:cmp_dhad_lenet_100}
}
\vspace{-1mm}
\subfigure[CNN \& CIFAR-100]{
\includegraphics[width=0.22\linewidth]{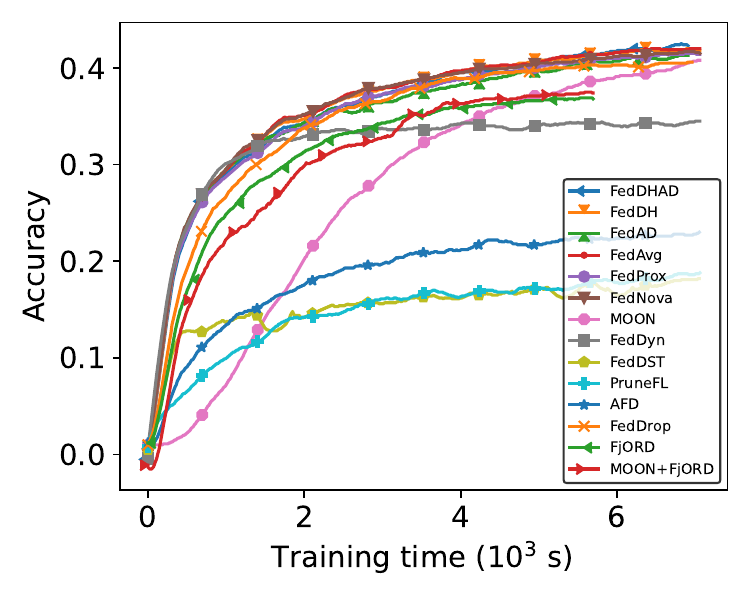}
\label{fig:cmp_dhad_cnn_100}
}
% \vspace{-3mm}
\caption{The accuracy and training time with \TheName{} and various baseline methods for LeNet and CNN with CIFAR-10 and CIFAR-100.}
\vspace{-6mm}
\label{fig:cmp_dhad_cifar}
\end{figure*}

\begin{figure*}[t]
\centering
\subfigure[LeNet \& SVHN]{
\includegraphics[width=0.22\linewidth]{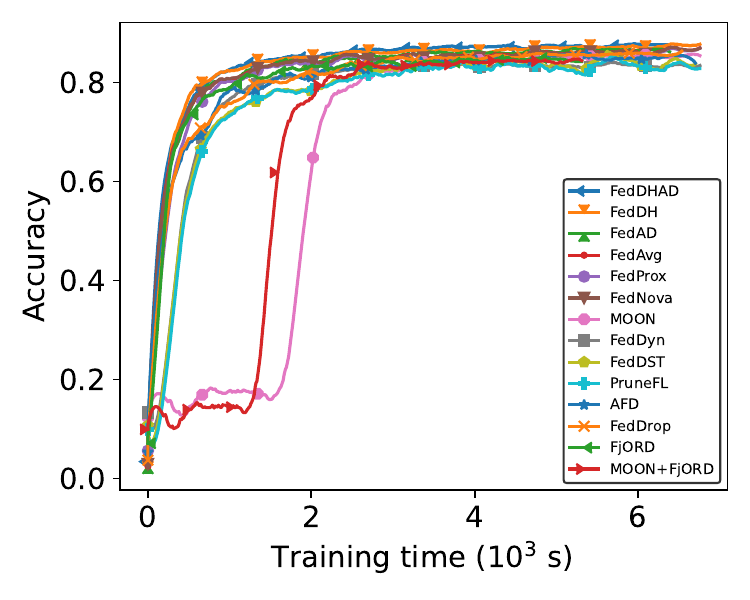}
\label{fig:cmp_dhad_lenet_svhn}
}
\vspace{-3mm}
\subfigure[CNN \& SVHN]{
\includegraphics[width=0.22\linewidth]{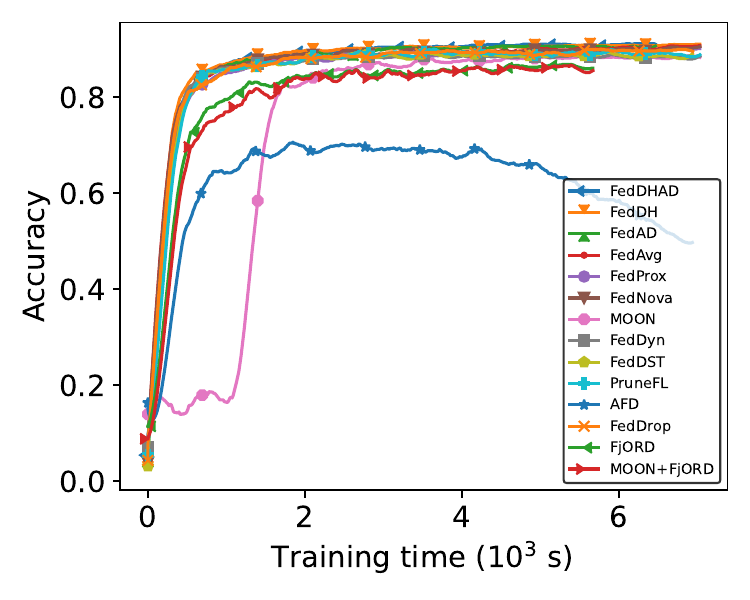}
\label{fig:cmp_dhad_cnn_svhn}
}
\subfigure[VGG \& SVHN]{
\includegraphics[width=0.22\linewidth]{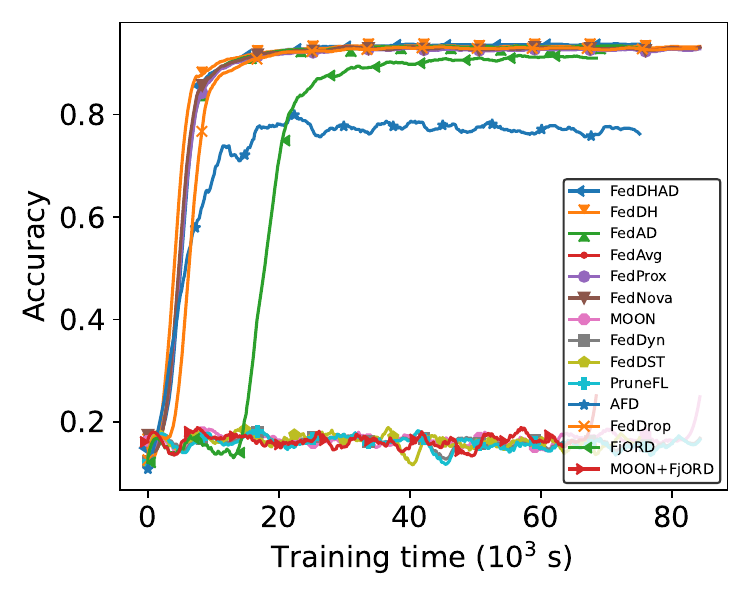}
\label{fig:cmp_dhad_vgg_svhn}
}
% \vspace{-3mm}
\subfigure[ResNet \& TinyImageNet]{
\includegraphics[width=0.22\linewidth]{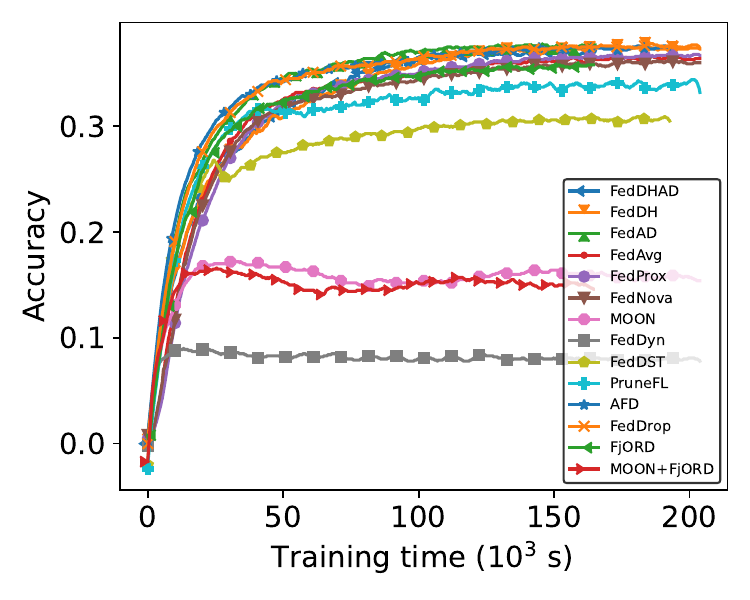}
\label{fig:cmp_dhad_resnet_tinyimagenet}
}
\caption{The accuracy and training time with \TheName{} and various baseline methods for LeNet, CNN, and VGG with SVHN and ResNet with TinyImageNet.}
\vspace{-6mm}
\label{fig:cmp_dhad_svhn_tinyimagenet}
%\vspace{-1mm}
\end{figure*}

%\begin{figure}[t]
%\centering
%\subfigure[LeNet]{
%\includegraphics[width=0.45\linewidth]{figures/FedDHAD_cifar10_lenet.pdf}
%}
% \vspace{-4mm}
%\subfigure[CNN]{
%\includegraphics[width=0.45\linewidth]{figures/FedDHAD_cifar10_cnn.pdf}
%}

%\subfigure[ResNet]{
%\includegraphics[width=0.47\linewidth]{fig/resnet/cifar100_cmp_share_time_p5_NonIID_Accuracy.pdf}
%}
% \vspace{-4mm}
%\subfigure[ResNet]{
%\includegraphics[width=0.47\linewidth]{fig/resnet/cifar100_cmp_share_time_p10_NonIID_Accuracy.pdf}
%}

%\caption{The accuracy and training time with \TheName{}, \TheNameNova{}, and various baseline methods.}
% \vspace{-6mm}
%\label{fig:cmp_DHAD}
%\end{figure}

As shown in Tables \ref{tab:cmp_DHAD_10} and \ref{tab:cmp_DHAD_100}, \TheName{} leads to high accuracy and fast convergence. The accuracy of \TheName{} is significantly higher than FedAvg (up to 6.7\%% 3.2\% for LeNet, 3.4\% for CNN
), FedProx (up to 4.5\%%2.9\% for LeNet, 3.1\% for CNN
), FedNova (up to 5.1\%%3.3\% for LeNet, 3.3\% for CNN
), MOON (up to 74.5\%%1.1\% for LeNet, 6.8\% for CNN
), FedDyn (up to 76.8\%%1.1\% for LeNet, 6.8\% for CNN
), FedDST (up to 78.6\%
), PruneFL (up to 61.5\%
 ), % FedDUAP (up to 21.5\%
% ),
AFD (up to 19.5\%%1.1\% for LeNet, 6.8\% for CNN
), FedDrop (up to 4.8\%%1.1\% for LeNet, 6.8\% for CNN
), FjORD (up to 4.7\%) and the combination of MOON and FjORD (up to 8.9\%) for LeNet, CNN, VGG, and ResNet, because of dynamic adjustment of control parameters and adaptive dropout operations. 
%In addition, the advantage of \TheName{} is still significant compared with FedAvg (up to ***\%
%), FedProx (up to ***\%), FedNova (up to ***\%), MOON (up to ***\%), FedDyn (up to ***\%), AFD (up to ***\%), FedDrop (up to ***\%), FjORD (up to ***\%) for both VGG and ResNet.
In addition, \TheName{} corresponds to the shortest training time 
(up to 1.69 times shorter %4.0\% for LeNet, 43.1\% for CNN
compared with FedAvg, 43.3\% %7.7\% for LeNet, 1.51 times for CNN
compared with FedProx, 42.0\% %40.8\% for LeNet, 1.76 times for CNN
compared with FedNova, up to 2.02 times shorter %39.9\% for LeNet, 2.57 times for CNN
compared with MOON, 45.9\% %\ for LeNet, 37.6\% for CNN
compared with FedDyn,
40.7\% compared with FedDST, 29.4\% compared with PruneFL, % 48.0\% compared with FedDUAP, 
1.85 times %\ for LeNet, 37.6\% for CNN
compared with AFD, and 1.91 times %42.1\% for LeNet, 2.17 times for CNN
compared with FedDrop, 32.5\% compared with FjORD, and 51.4\% compared with the combination of MOON and FjORD) to achieve target accuracy for LeNet, CNN, VGG, and ResNet. Although it corresponds to a slightly longer time to achieve the target time compared with AFD (0.47 times longer), \TheName{} leads to higher accuracy (0.2\% compared with AFD) with ResNet.  \ma{} In addition, \TheName{} incurs smaller computation costs (up to 14.9\% compared with FedAvg, FedProx, FedNova, FedDyn, and PruneFL, 15.0\% compared with MOON, 10.1\% compared with AFD, and 9.8\% compared with FedDrop) than those presented above, while corresponding to slightly higher computation costs (up to 28.6\%) than FjORD. \ma{In addition, the total training time of \TheName{} is shorter than that of \TheDropoutName{} (up to 28.7\%). This is primarily because of the dynamic adjustment of the weights of local models, which corresponds to higher accuracy. \TheName, by incorporating \TheHetName{}, adjusts the weights of local models during the model aggregation process based on the non-IID degrees of the heterogeneous data, achieving more efficient model aggregation and global model updates compared with \TheDropoutName{}. However, \TheDropoutName{} only performs the dropout rate update. The dropout rate update operation adjusts the dropout rates based on the rank or weight of filters and neurons, which can avoid performance degradation and exist in \TheName{} as well. Thus, the combination of \TheHetName{} and \TheDropoutName{} leads to a more efficient optimization process in \TheName{} compared to \TheDropoutName{} alone. }

We carry out experiments to evaluate the performance of \TheName{} with various network bandwidth, various non-IID degrees of data, and different numbers of devices. When the network connection becomes worse, the advantages of \TheName{} become significant (from 64.9\% to 82.7\%) because of communication reduction brought by the adaptive dropout. As it can exploit the non-IID degrees of the data on each device to adjust the weights of multiple models, \TheName{} achieves superior accuracy (from 6.9\% to 10.7\%) even with the data of high non-IID degrees. Furthermore, \TheName{} incurs 13.5\% higher accuracy and 8.6\% shorter training time to achieve target accuracy with different numbers of devices and the significance becomes obvious with a large scale of devices.

\begin{figure*}[t]
\centering
\subfigure[LeNet \& CIFAR-10]{
\includegraphics[width=0.22\linewidth]{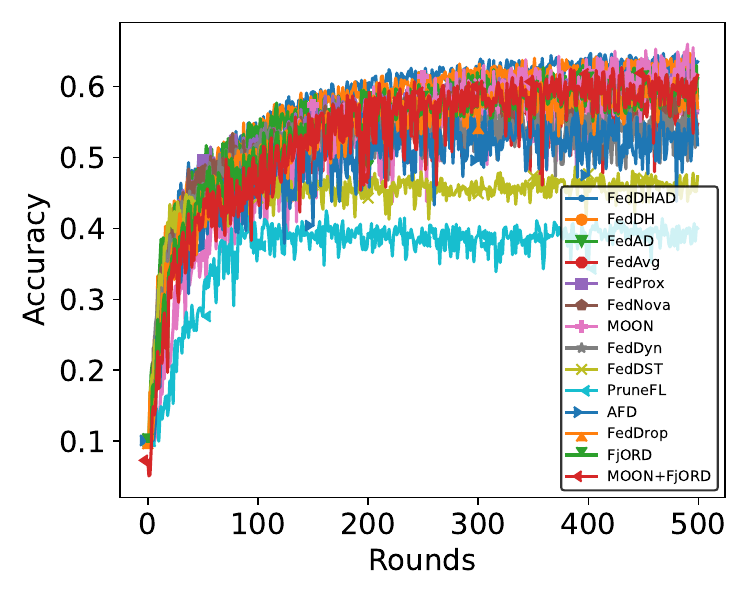}
\label{fig:acc_rounds_dhad_lenet_10}
}
\vspace{-1mm}
\subfigure[CNN \& CIFAR-10]{
\includegraphics[width=0.22\linewidth]{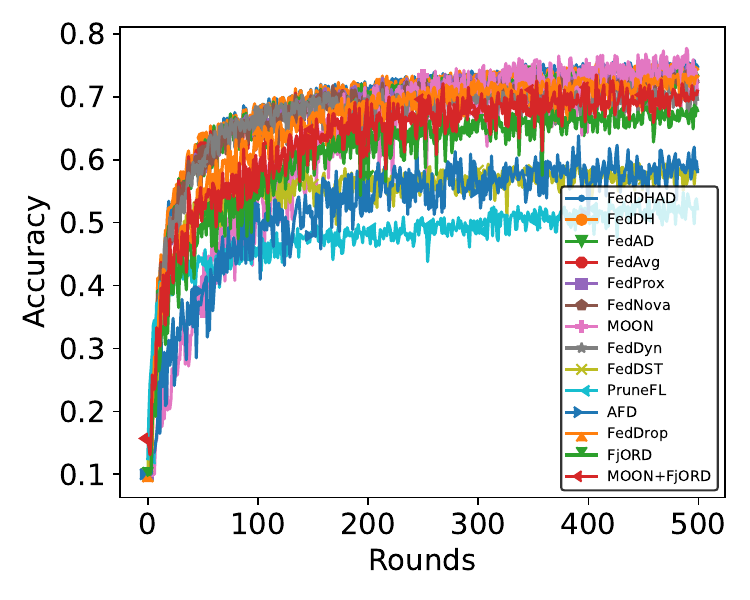}
\label{fig:acc_rounds_dhad_cnn_10}
}
\vspace{-1mm}
\subfigure[LeNet \& CIFAR-100]{
\includegraphics[width=0.22\linewidth]{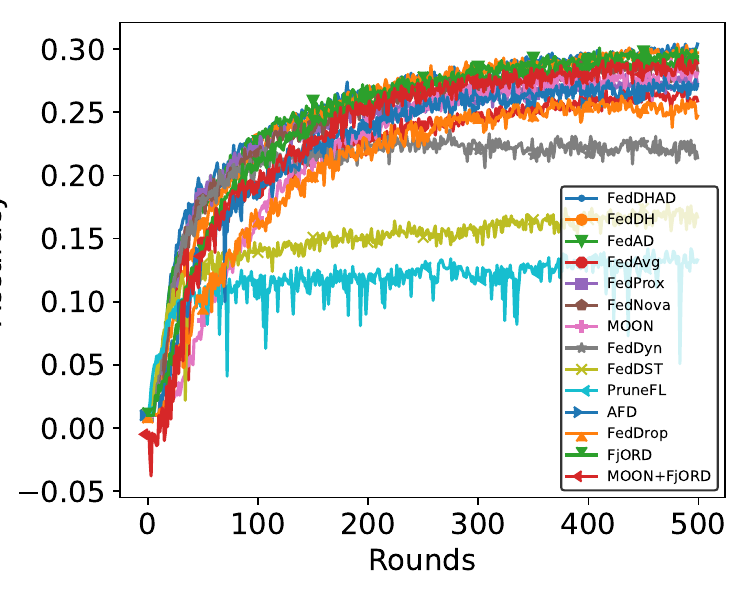}
\label{fig:acc_rounds_dhad_lenet_100}
}
\vspace{-1mm}
\subfigure[CNN \& CIFAR-100]{
\includegraphics[width=0.22\linewidth]{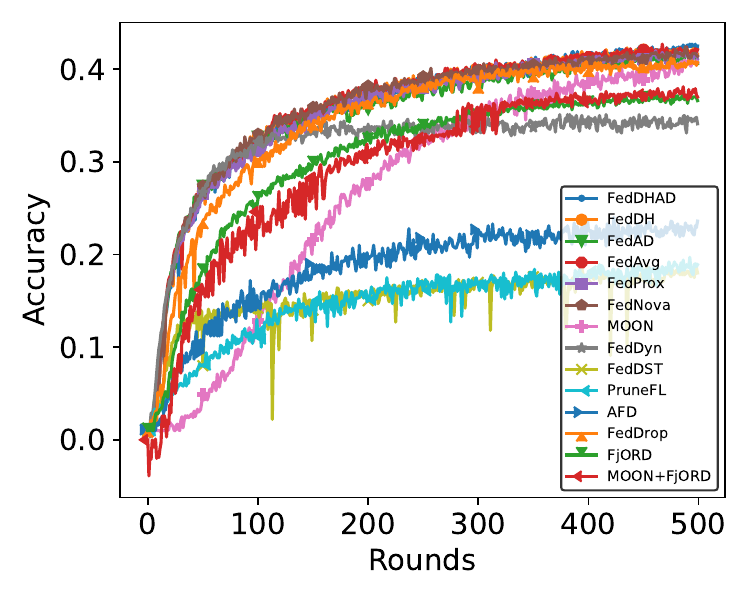}
\label{fig:acc_rounds_dhad_cnn_100}
}
\vspace{-1mm}
\subfigure[LeNet \& SVHN]{
\includegraphics[width=0.22\linewidth]{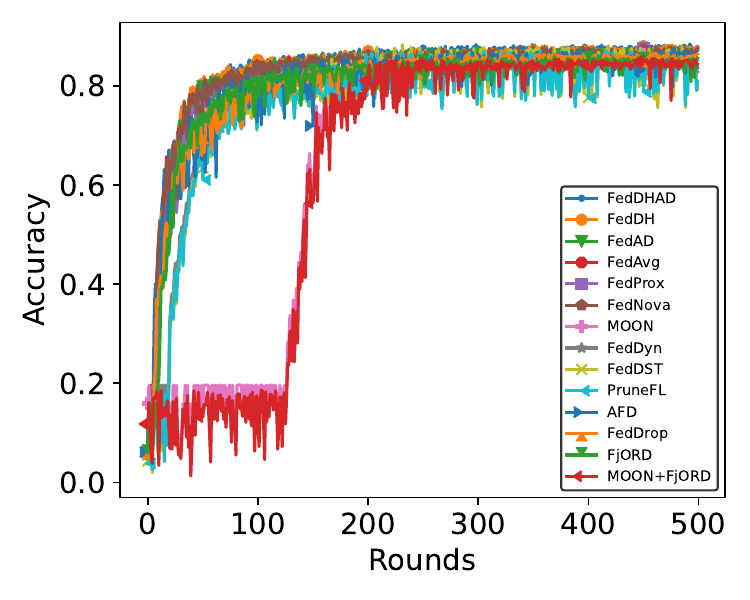}
\label{fig:acc_rounds_dhad_lenet_svhn}
}
\vspace{-1mm}
\subfigure[CNN \& SVHN]{
\includegraphics[width=0.22\linewidth]{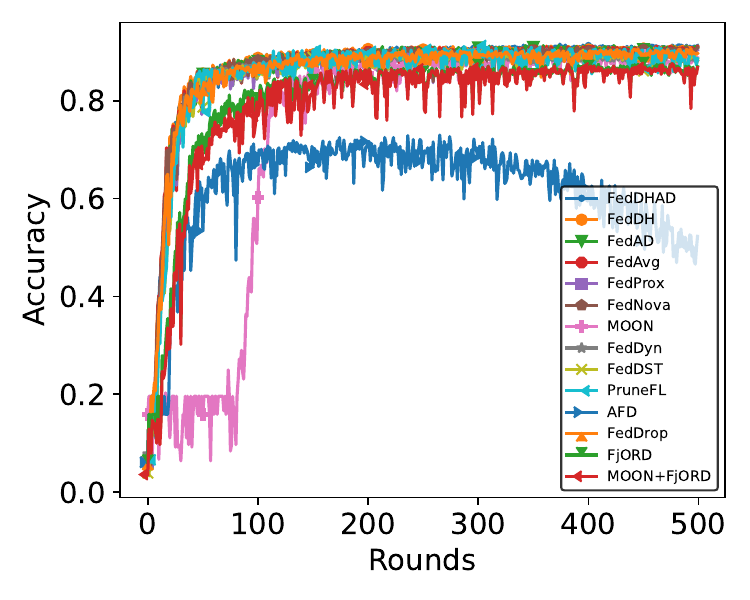}
\label{fig:acc_rounds_dhad_cnn_svhn}
}
\vspace{-1mm}
\subfigure[VGG \& SVHN]{
\includegraphics[width=0.22\linewidth]{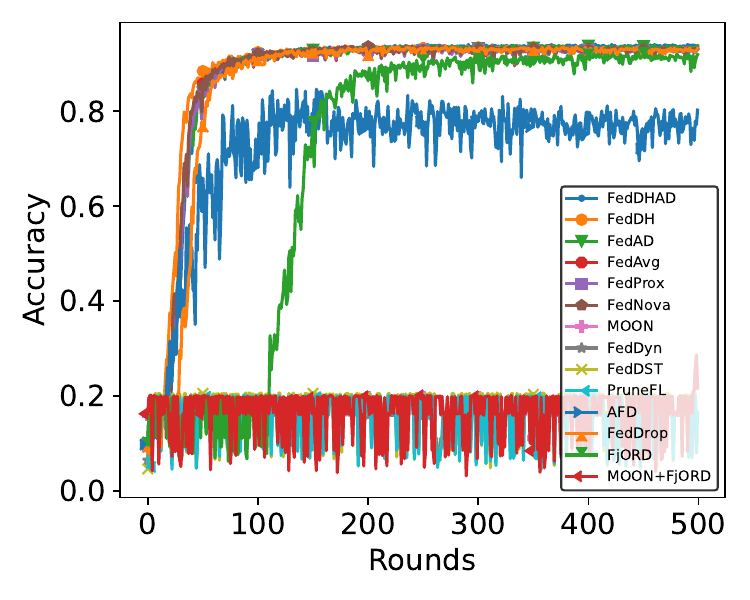}
\label{fig:acc_rounds_dhad_vgg_svhn}
}
\vspace{-1mm}
\subfigure[ResNet \& TinyImageNet]{
\includegraphics[width=0.22\linewidth]{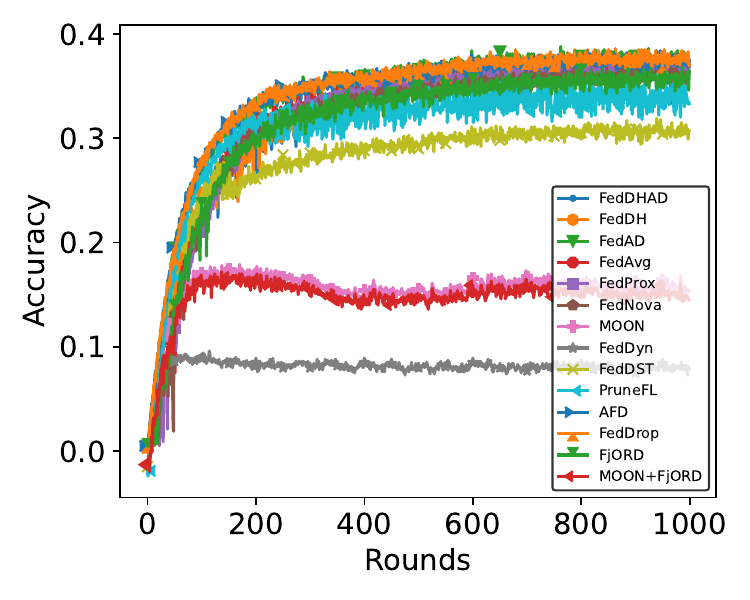}
\label{fig:acc_rounds_dhad_resnet_tinyimagenet}
}
% \vspace{-2mm}
\caption{The accuracy and training rounds with \TheName{} and various baseline methods for LeNet and CNN with CIFAR-10 and CIFAR-100, LeNet, CNN, and VGG with SVHN and ResNet with TinyImageNet.}
\vspace{-6mm}
\label{fig:acc_rounds_dhad}
\end{figure*}

\subsection{Evaluation with various Environments}

In this section, we give our experimental results performed with various environments
along three dimensions: network bandwidth, statistical data heterogeneity and scalability in number of devices.

\subsubsection{Network Bandwidth}

While devices have limited network connection, we analyze the performance of \TheName{} with different bandwidths. As shown in Figure \ref{fig: network_impact}, when the bandwidth becomes modest, \TheName{} corresponds to significantly higher training speed (up to 64.9\%) compared with baseline approaches. We observe \TheName{} corresponds to higher accuracy (up to 10.7\%) compared with baseline approaches with the modest bandwidth as well. When the network connection becomes worse, the advantages of \TheName{} become significant (from 64.9\% to 82.7\%) because of communication reduction brought by the adaptive dropout.

\begin{figure*}[t]
\centering
\subfigure[Time to target accuracy (0.54) under various bandwidths. ``Com.'' represents "Communication".]{
\includegraphics[width=0.3\linewidth]{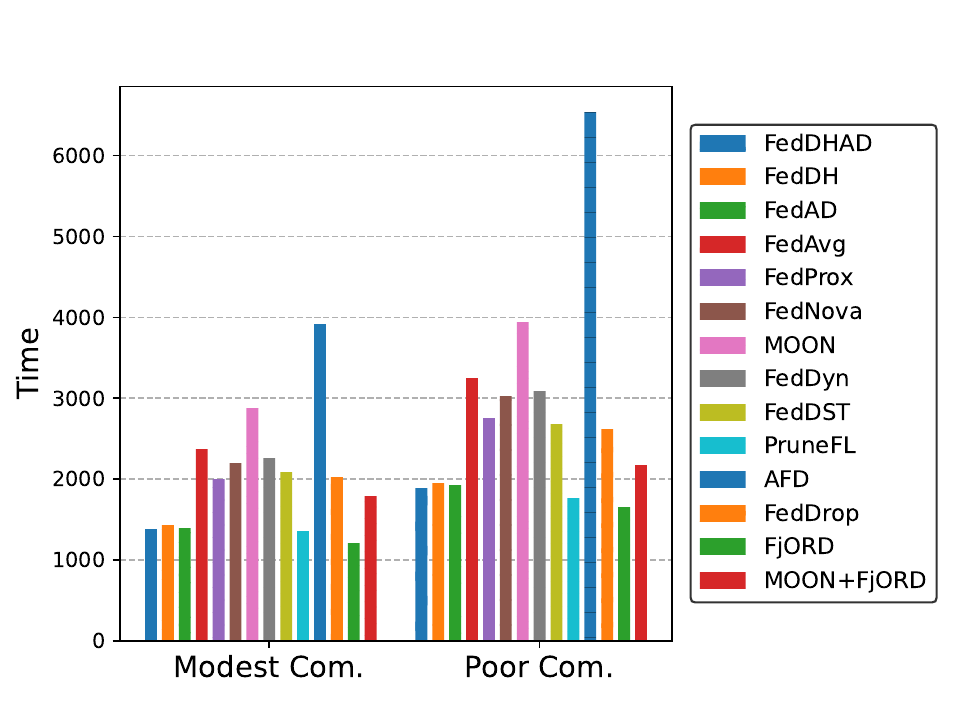}
\label{fig: network_impact}
}
% \vspace{-2mm}
\subfigure[The accuracy under various device heterogeneity. ``Het.'' represents "Heterogeneity".]{
\includegraphics[width=0.3\linewidth]{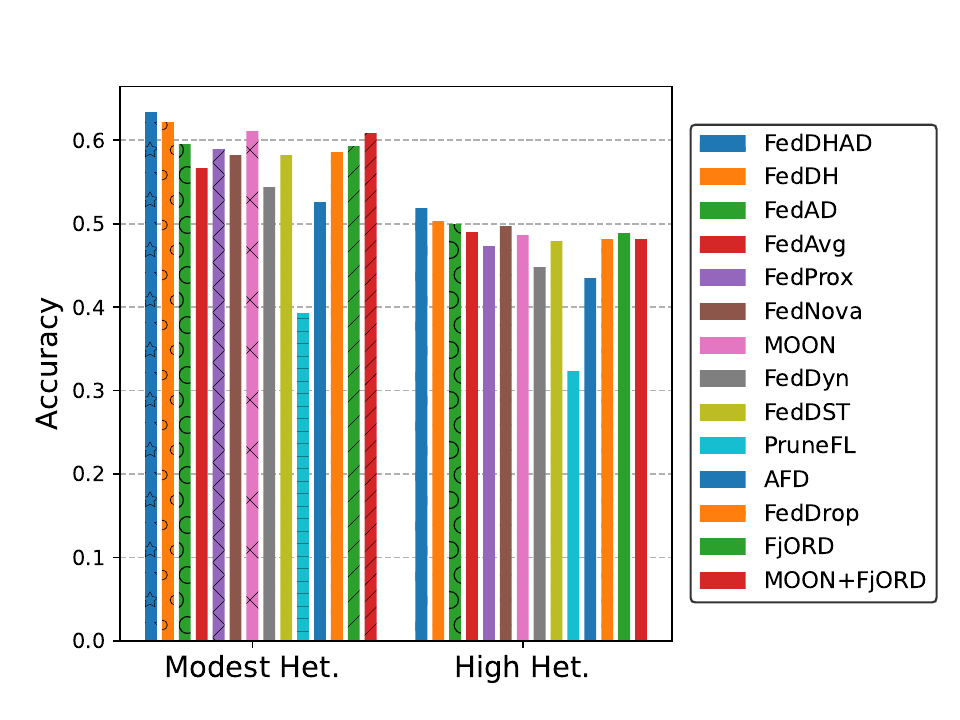}
\label{fig: heterogeneity_impact}
}
\vspace{-2mm}
\subfigure[The accuracy with different numbers of devices.]{
\includegraphics[width=0.3\linewidth]{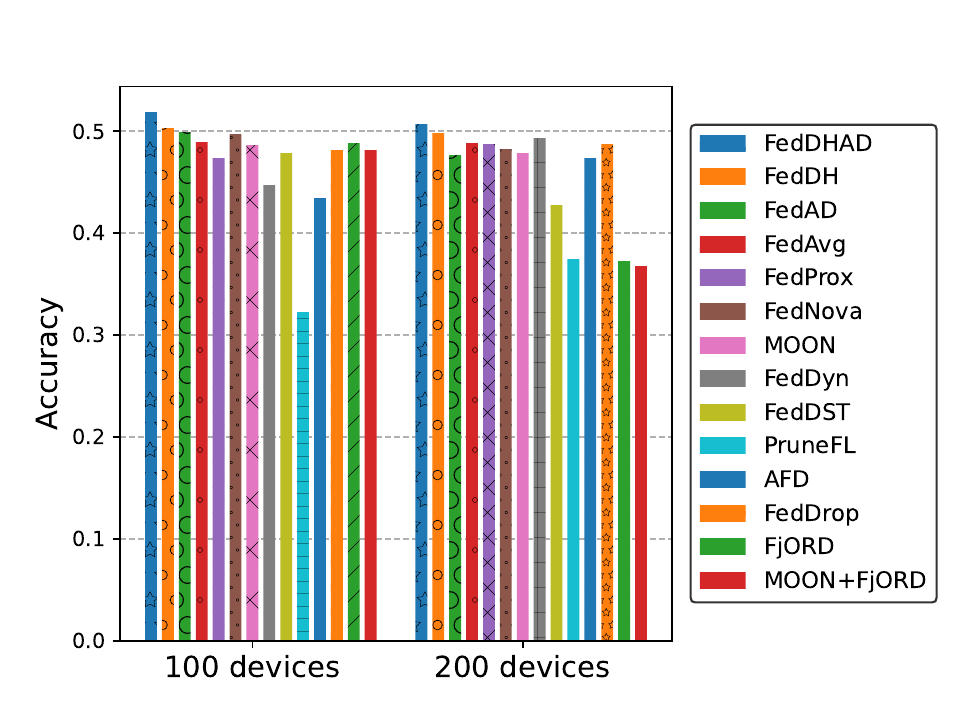}
\label{fig: device_impact}
}
\vspace{-2mm}
\caption{Performance of \TheName{} and various baseline methods with LeNet on CIFAR-10.}
\vspace{-5mm}
\label{fig:various_env}
\end{figure*}

\subsubsection{Statistical Data Heterogeneity}

As shown in Figure \ref{fig: heterogeneity_impact}, when the device data is of high heterogeneity, i.e., high non-IID degree data, \TheName{} significantly outperforms other baseline approaches in terms of accuracy (up to 6.9\%). In addition, we find \TheName{} leads to a shorter training time (up to 52.2\%) to achieve the target accuracy compared with baseline approaches with the data of high non-IID degree, because of dynamic adjustment of control parameters and adaptive dropout operations. When the device data are highly heterogeneous, the advantage of \TheName{} corresponds to higher accuracy becomes obvious (from 6.9\% to 10.7\%). As it can exploit the non-IID degrees of the distributed data on each device to adjust the weights of multiple models, \TheName{} achieves superior accuracy with the data of high non-IID degrees.

\subsubsection{Scalability}

To show the scalability of \TheName{}, we carry out experiments with a large number (200) of devices with LeNet on CIFAR-10. The results are shown in Figure \ref{fig: device_impact}, \TheName{} significantly outperforms other baseline approaches in terms of accuracy (up to 13.5\%). In addition, we find \TheName{} incurs a shorter training time (up to 8.6\%) to achieve target accuracy with 200 devices compared with that of 100. Furthermore, we study the effect of the number of devices on multiple approaches as shown in Figure \ref{fig: device_impact}. As the number of devices increases, the amount of local data decreases. We observe that the accuracy of baseline approaches decreases significantly when increasing the number of devices (up to 11.6\%). However, The advantages of \TheName{} become significant on a large-scale setting with small data in the device.

\subsection{Evaluation of Single Methods}
In this section, we further evaluate the single methods that contribute to the superiority of our approach. 
We first present the evaluation of our data distribution estimation method. 
Then, we give a comparison between \TheHetName{} and a static method. Afterward, we analyze the computation cost of \TheHetName{} (to update $Q$ in Formula \ref{eq:bilevel}) on the server. Finally, we provide the ablation study corresponding to the contribution of \TheHetName{} and \TheDropoutName{} to \TheName{}.
%Ji: contribution to what?
%JL: added.

\subsubsection{Data Distribution Estimation}
We carried out an experiment to verify the difference between \TheHetName{} and the data distribution estimation method without transferring meta information, i.e., \TheHetNameE{}. 
As shown in Figures \ref{fig:cmp_dhe_lenet_10} and \ref{fig:cmp_dhe_cnn_10}, we find that \TheHetName{} and \TheHetNameE{} lead to similar accuracy (with the difference less than 0.8\% for both LeNet and CNN). Thus, when the meta information is not allowed to be transferred from devices to the server for high privacy requirements, we can exploit \TheHetNameE{} instead of \TheHetName{}.

\subsubsection{Comparison with Static JS Divergence}
To show that static JS divergence (FedJS) cannot well represent the non-IID degree, we carried out an experiment with LeNet and CNN on CIFAR-10, which reveals the difference between \TheHetName{} and FedJS. Figures \ref{fig:cmp_js_lenet_10} and \ref{fig:cmp_js_cnn_10} show that \TheHetName{} significantly outperforms FedJS in terms of accuracy (up to 2.9\% for LetNet and 0.9\% for CNN).

\subsubsection{Server Computation Overhead}
As \TheName{} involves server computation, we quantify its overhead in terms of time. As shown in Table \ref{tab:server_burden},
the overheads of \TheHetName{} are up to 10.4\%, which is quite small. 
In addition, the overheads of \TheName{} are up to 15.6\%, which corresponds to higher accuracy, shorter time to achieve target accuracy, and smaller computation costs.

\begin{figure*}[t]
\centering
\subfigure[LeNet CIFAR-10]{
\includegraphics[width=0.22\linewidth]{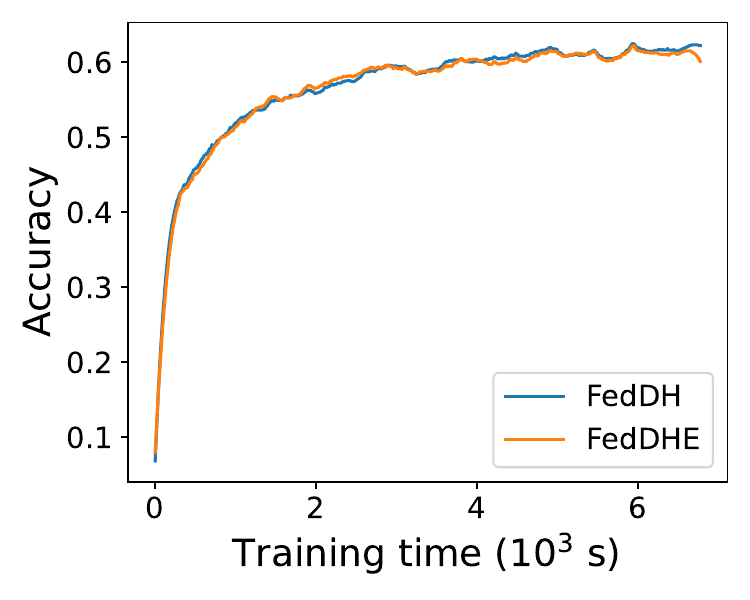}
\label{fig:cmp_dhe_lenet_10}
}
\vspace{-1mm}
\subfigure[CNN CIFAR-10]{
\includegraphics[width=0.22\linewidth]{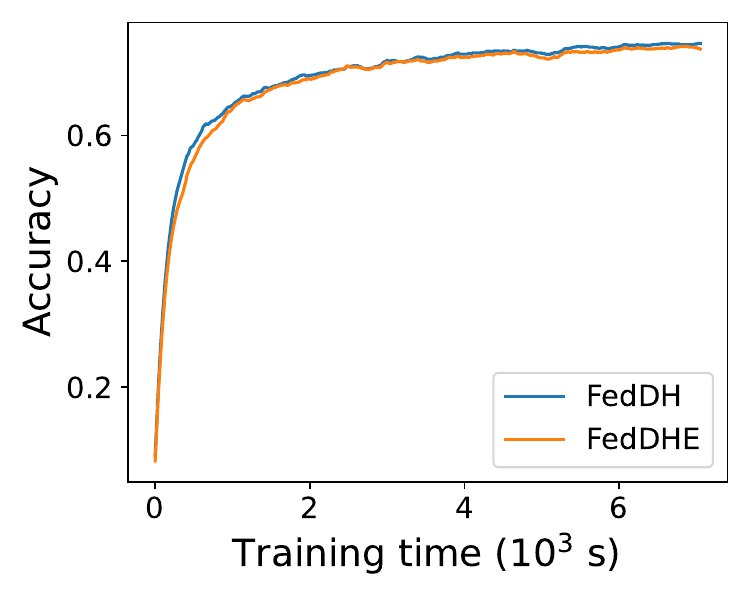}
\label{fig:cmp_dhe_cnn_10}
}
\vspace{-1mm}
\subfigure[LeNet CIFAR-10]{
\includegraphics[width=0.22\linewidth]{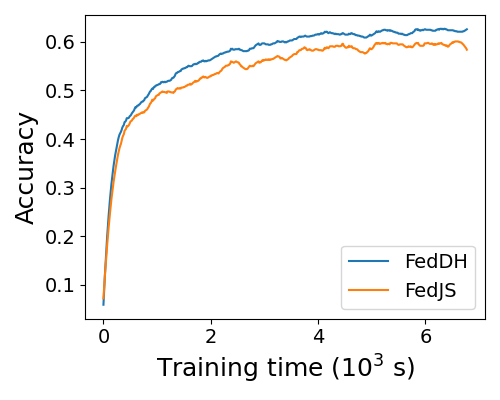}
\label{fig:cmp_js_lenet_10}
}
\vspace{-1mm}
\subfigure[CNN CIFAR-10]{
\includegraphics[width=0.22\linewidth]{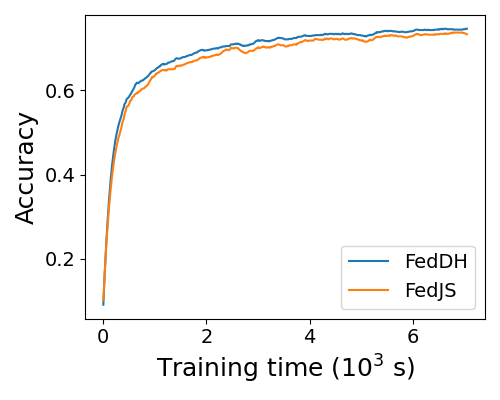}
\label{fig:cmp_js_cnn_10}
}
% \vspace{-3mm}
\caption{The accuracy and training time with \TheHetName{}, \TheHetNameE{}, and FedJS for LeNet and CNN on CIFAR-10.}
\vspace{-7mm}
\label{fig:cmp_js}
\end{figure*}

\subsubsection{Ablation Study}
In this section, we present the advantages of \TheHetName{} and \TheHetName{}, corresponding to the specific contribution to \TheName{}.
As shown in Tables \ref{tab:cmp_DHAD_10} and \ref{tab:cmp_DHAD_100}, \TheName{} can achieve higher accuracy and shorter training time to achieve the target accuracy compared with \TheHetName{} (up to 1.1\% in terms of accuracy and 14.7\% in terms of training time) and \TheDropoutName{} (up to 3.8\% in terms of accuracy and 28.7\% in terms of training time). As \TheName{} carries out the dynamic adjustment of the weights and the adaptive dropout operations at the same time, it leads to the shortest training time (up to 26.8\%) to achieve target accuracy (except for TinyImageNet with ResNet) and corresponds to the same computation cost or slightly higher (up to 0.5\%) for CIFAR-10, CIFAR-100, and SVHN, compared with \TheDropoutName{}. \TheHetName{} corresponds to higher accuracy compared with \TheDropoutName{} (up to 2.6\%) and \TheDropoutName{} incurs smaller computation costs (up to 15.4\%). \TheName{} corresponds to slightly higher computation costs compared with \TheDropoutName{} (up to 0.5\%) because of \TheHetName{}. With ResNet and TinyImageNet, \TheName{} corresponds to the highest accuracy while \TheDropoutName{} incurs the shortest time to achieve the target accuracy and computational costs. When tight training budget is required and big models are exploited, \TheDropoutName{} can be exploited. Otherwise, \TheName{} can be selected. %for superb accuracy.

\section{Conclusion}
\label{sec:conclusion}

In this paper, we proposed \TheName{}, an FL framework that simultaneously addresses statistical data heterogeneity and device heterogeneity. 
We proposed two novel methods: Dynamic Heterogeneous model aggregation (\TheHetName{}) and Adaptive Dropout (\TheDropoutName{}).
\TheHetName{} exploits the non-IID degrees of devices to dynamically adjust the weights of each local model within the model aggregation process.
\TheDropoutName{} adapts the dropout operation to each device while reducing communication and computation costs and achieving a good balance between efficiency and effectiveness.
By combining these two methods, \TheName{} is able to adaptively calculate the dropout rates in each round with good  load balancing, thus addressing the training efficiency problem, and dynamically adjusts the weights of the models with removed neurons, thus improving model accuracy.

The extensive experimental results reveal that \TheName{} significantly outperforms the state-of-the-art baseline approaches in terms of accuracy (up to 6.7\% higher), efficiency (up to 2.02 times faster), and computation cost (up to 15.0\% smaller). The experimental evaluation demonstrates as well that \TheName{} has significant advantages in various environments, e.g., modest network (\TheName{} is robust to bandwidth), high non-IID degree (\TheName{} is effective with heterogeneous data) and large number of devices (\TheName{} has excellent scalability). 

\ma{
Large Language Models (LLMs) have shown great potential in numerous tasks due to their large-scale pre-trained architectures. The training process of LLMs typically involves updating a large number of parameters \cite{liu2024fisher,che2023federated}, limiting the use of FL techniques in real-world scenarios. Future directions could explore using \TheName{} principles, focusing on efficient model updates and adaptive dropout tailored for LLMs. Additionally, future directions could explore improving the adaptivity of the \TheName{} framework by incorporating fine-grained control over dropout and model aggregation weights for extreme heterogeneity cases. By integrating personalized federated learning approaches, we can further ensure individual client needs are addressed, enhancing model performance across diverse data distributions. In addition, \TheName{} relies on central server for model aggregation. Future directions could extend to decentralized federated learning architectures, where model aggregation is performed collaboratively by participating devices. This shift could enhance robustness, reduce communication latency, and increase the resilience of system to server-side attacks or failures. By addressing these future directions, we aim to further the capabilities and applicability of federated learning, paving the way for more efficient, decentralized, and secure machine learning systems at scale.
}

%\section*{Acknowledgments}
%This should be a simple paragraph before the References to thank those individuals and institutions who have supported your work on this article.

% {\appendix[Proof of the Zonklar Equations]
% Use $\backslash${\tt{appendix}} if you have a single appendix:
% Do not use $\backslash${\tt{section}} anymore after $\backslash${\tt{appendix}}, only $\backslash${\tt{section*}}.
% If you have multiple appendixes use $\backslash${\tt{appendices}} then use $\backslash${\tt{section}} to start each appendix.
% You must declare a $\backslash${\tt{section}} before using any $\backslash${\tt{subsection}} or using $\backslash${\tt{label}} ($\backslash${\tt{appendices}} by itself
% starts a section numbered zero.)}
\section*{Appendix}

In this section, we present the theoretical convergence proof, the correlation between JS divergence and non-IID degree, and the limitations of \TheName{}.

\subsection*{Convergence Proof}

We first introduce the assumptions and then present the convergence theorems with an upper bound.

\begin{proof}
%See details in Appendix.
First we have
\vspace{-2mm}
\begin{align*}
\vspace{-2mm}
    \parallel \bar{w}_{t,h+1} - {w}^* \parallel %\\
    %= &\parallel \bar{w}_{t,h} - \eta_{t,h} g_{t, h} - {w}^* - \eta_{t,h} \bar{g}_{t,h} + \eta_{t,h} \bar{g}_{t, h} \parallel ^2 \\
    = &\underbrace{\parallel \bar{w}_{t,h} - \eta_{t,h} \bar{g}_{t,h} - {w}^* \parallel ^2}_{A_1} + \underbrace{\parallel \eta_{t,h} \bar{g}_{t,h} - \eta_{t,h} g_{t, h} \parallel ^2}_{A_2} \\
    &+ 2\underbrace{\langle\bar{w}_{t,h} - \eta_{t,h} \bar{g}_{t, h} - {w}^*, \eta_{t,h} \bar{g}_{t, h} - \eta_{t,h} g_{t, h}\rangle}_{A_3}
\vspace{-4mm}
\end{align*}
Note that $E[A_3] = 0$ when devices are with unbiased sampling, bound $A_1$:
\vspace{-2mm}
\begin{align*}
\vspace{-4mm}
    A_1 = \parallel \bar{w}_{t,h} - {w}^* \parallel ^2 + \underbrace{\parallel \eta_{t,h} \bar{g}_{t,h} \parallel^2}_{B_1} - \underbrace{2\eta_{t,h} \langle\bar{w}_{t,h} - {w}^*, \bar{g}_{t, h} \rangle}_{B_2}
\vspace{-2mm}
\end{align*}
By the convexity of $\parallel \cdot \parallel^2$ (Jensen's inequality) and L-smoothness of $F_k$, 
%\vspace{-2mm}
\begin{align*}
%\vspace{-3mm}
    B_1 %&= {\eta_{t,h}}^2\parallel \bar{g}_{t, h}\parallel^2  \\
    %&= {\eta_{t,h}}^2\parallel \sum_{k \in S_t} p_{k,t} \nabla F_k(w^{k}_{t,h})\parallel^2 \\ 
    %&\leq {\eta_{t,h}}^2\sum_{k \in S_t}p_{k,t} \parallel\nabla F_k(w^{k}_{t,h}) \parallel^2 \\
    %&
    \leq 2L{\eta_{t,h}}^2 \sum_{k \in S_t}p_{k,t} (F_k(w^{k}_{t,h}) - E[F_k({w_k}^*)]) 
\end{align*}
Next, we bound $B_2$:
\vspace{-2mm}
\begin{align*}
\vspace{-4mm}
    B_2 %= &-2\eta_{t,h}\langle\bar{w}_{t,h} - {w}^*, \bar{g}_{t, h}\rangle \\
    %= &-2\eta_{t,h} \sum_{k \in S_t}p_{k,t} \langle \bar{w}_{t,h} - {w}^*, \nabla F_k(w^{k}_{t,h}) \rangle \\
    = &-2\eta_{t,h} \sum_{k \in S_t}p_{k,t} \langle \bar{w}_{t,h} - w^{k}_{t,h}, \nabla F_k(w^{k}_{t,h}) \rangle \\
    &-2\eta_{t,h} \sum_{k \in S_t}p_{k,t} \langle w^{k}_{t,h} - {w}^*, \nabla F_k(w^{k}_{t,h}) \rangle
\end{align*}
By Cauchy-Schwarz inequality and AM-GM inequality, we have
\vspace{-2mm}
\begin{align*}
\vspace{-2mm}
    &-2\langle \bar{w}_{t,h} - w^{k}_{t,h}, \nabla F_k(w^{k}_{t,h}) \rangle \\
    \leq~&\frac{1}{\eta_{t,h}}\parallel\bar{w}_{t,h} - w^{k}_{t,h}\parallel^2 + \eta_{t,h} \parallel \nabla F_k(w^{k}_{t,h}) \parallel^2
\end{align*}
By $\mu$-convexity of $F_k$, we have,
\vspace{-2mm}
\begin{align*}
\vspace{-2mm}
    &-\langle w^{k}_{t,h} - {w}^*, \nabla F_k(w^{k}_{t,h}) \rangle \\
    \le &- (F_k(w^{k}_{t,h}) - F_k({w}^*)) - \frac{\mu}{2} \parallel \bar{w}_{t,h} - {w}^* \parallel^2
\end{align*}
plug $B_1, B_2$ into $A_1$:
\vspace{-2mm}
\begin{align*}
\vspace{-2mm}
    A_1 %\\
    %\leq~&\parallel \bar{w}_{t,h} - {w}^* \parallel^2 + 2L{\eta_{t,h}}^2 \sum_{k \in S_t}p_{k,t} \bigl(F_k(w^{k}_{t,h}) - E[F_k({w_k}^*)]\bigr) \\
    %&+ \eta_{t,h} \sum_{k \in S_t} p_{k,t} \bigl(\frac{1}{\eta_{t,h}}\parallel\bar{w}_{t,h} - w^{k}_{t,h}\parallel^2 + \eta_{t,h} \parallel \nabla F_k(w^{k}_{t,h}) \parallel^2 \bigr) \\
    %&- 2\eta_{t,h} \sum_{k \in S_t}p_{k,t} \bigl(F_k(w^{k}_{t,h}) - F_k({w}^*) + \frac{\mu}{2} \parallel \bar{w}_{t,h} - {w}^* \parallel^2 \bigr) \\
    %=~&(1-\mu \eta_{t,h})\parallel \bar{w}_{t,h} - {w}^* \parallel ^2 + \sum_{k \in S_t} p_{k,t} \parallel \bar{w}_{t,h} - w^{k}_{t,h} \parallel^2 \\
    %&+ {\eta_{t,h}}^2 \sum_{k \in S_t}p_{k,t}\parallel \nabla F_k(w^{k}_{t,h}) \parallel^2\\
    %&+2L{\eta_{t,h}}^2 \sum_{k \in S_t} p_{k,t} \bigl(F_k(w^{k}_{t,h}) - E[F_k({w_k}^*)]\bigr) \\
    %&-2\eta_{t,h} \sum_{k \in S_t} p_{k,t} \bigl(F_k(w^{k}_{t,h}) - F_k({w}^*)\bigr) \\
    \leq~&(1-\mu \eta_{t,h})\parallel \bar{w}_{t,h} - {w}^* \parallel ^2 + \sum_{k \in S_t} p_{k,t} \parallel \bar{w}_{t,h} - w^{k}_{t,h} \parallel^2 \\
    &+\underbrace{4L{\eta_{t,h}}^2 \sum_{k \in S_t} p_{k,t} \bigl(F_k(w^{k}_{t,h}) - E[F_k({w_k}^*)]\bigr)}_{C_1} \\
    &- \underbrace{2\eta_{t,h} \sum_{k \in S_t} p_{k,t} \bigl(F_k(w^{k}_{t,h}) - F_k({w}^*)\bigr)}_{C_2} 
\end{align*}
where the last inequality results from $L$-smoothness of $F_k$. Next, we quantify the degree of non-iid:
\vspace{-2mm}
\begin{equation}
\label{eq:non-iid-def}
    \Gamma_{k,t} = E[F_k({w}^*)] - E[F_k({w_k}^*)]
\end{equation}
where $k \in S_t $.
\vspace{-2mm}
\begin{align*}
\vspace{-4mm}
    C_1 - C_2 
    %=~&4L{\eta_{t,h}}^2 \sum_{k \in S_t} p_{k,t} \bigl(F_k(w^{k}_{t,h}) - E[F_k({w_k}^*)]\bigr) \\
    %&-2\eta_{t,h} \sum_{k \in S_t} p_{k,t} \bigl(F_k(w^{k}_{t,h}) - F_k({w}^*)\bigr) \notag \\
    %=~&4L{\eta_{t,h}}^2 \sum_{k \in S_t} p_{k,t} \bigl(F_k({w}^*)  - E[F_k({w_k}^*)]\bigr) \\
    %&-2\eta_{t,h} \sum_{k \in S_t} p_{k,t} \bigl(F_k(w^{k}_{t,h}) - F_k({w}^*)\bigr)  \notag \\
    %&+ 4L{\eta_{t,h}}^2 \sum_{k \in S_t} p_{k,t}\bigl( F_k(w^{k}_{t,h}) -  F_k({w}^*) \bigr) \notag \\
    %=~&4L{\eta_{t,h}}^2 \sum_{k \in S_t} p_{k,t} \bigl(F_k({w}^*)  - E[F_k({w_k}^*)]\bigr) \\
    %&- 2\eta_{t,h}(1-2L\eta_{t,h}) \sum_{k \in S_t} p_{k,t} \bigl(F_k(w^{k}_{t,h}) - F_k({w}^*)\bigr)  \notag \\
    \leq~&4L{\eta_{t,h}}^2 \sum_{k \in S_t} p_{k,t} \Gamma_{k,t} \\
    &- 2\eta_{t,h}(1-2L\eta_{t,h})\underbrace{\sum_{k \in S_t} p_{k,t}(F_k(w^{k}_{t,h}) - F_k({w}^*))}_{D}
\end{align*}
We define $\gamma = 2\eta_{t,h}(1-2L\eta_{t,h})$. If $\eta_{t,h} \le \frac{1}{4L}, \eta_{t,h} \le \gamma \le 2\eta_{t,h}$.
\vspace{-2mm}
\begin{align*}
\vspace{-2mm}
    D %=~&\sum_{k \in S_t} p_{k,t} (F_k(w^{k}_{t,h}) - F_k({w}^*)) \\ 
    %=~&\sum_{k \in S_t} p_{k,t} (F_k(w^{k}_{t,h}) - F_k(\bar{w}_{t,h})) \\
    %&+ \sum_{k \in S_t} p_{k,t} (F_k(\bar{w}_{t,h}) - F_k({w}^*)) \\
    %\ge~&\sum_{k \in S_t} p_{k,t} \langle \nabla F_k(\bar{w}_{t,h}), w^{k}_{t,h} - \bar{w}_{t,h} \rangle \\
    %&+ \sum_{k \in S_t} p_{k,t} (F_k(\bar{w}_{t,h}) - F_k({w}^*))\\
    %\ge~&-\frac{1}{2}\sum_{k \in S_t} p_{k,t} \bigl[\eta_{t,h} \parallel \nabla F_k(\bar{w}_{t,h}) \parallel^2 \\
    %&+ \frac{1}{\eta_{t,h}} \parallel w^{k}_{t,h} - \bar{w}_{t,h} \parallel ^2\bigr] \\
    %&+ \sum_{k \in S_t} p_{k,t} (F_k(\bar{w}_{t,h}) - F_k({w}^*)) \\
    \ge~&-\sum_{k \in S_t} p_{k,t} \bigl[\eta_{t,h} L (F_k(\bar{w}_{t, h}) - E[F_k({w_k}^*)]) \\
    &+ \frac{1}{2\eta_{t,h}} \parallel w^{k}_{t,h} - \bar{w}_{t, h} \parallel ^2\bigr]\\
    &+ \sum_{k \in S_t} p_{k,t} (F_k(\bar{w}_{t,h}) - F_k({w}^*))
\end{align*}
where the first inequality results from the convexity of $F_k(\cdot)$, the second and last inequality results from AM-GM inequality and L-smoothness. Next, we plug into $C_1-C_2$
\vspace{-1mm}
\begin{align*}
\vspace{-2mm}
    C_1-C_2 %\\
    %\leq~&4L{\eta_{t,h}}^2 \sum_{k \in S_t} p_{k,t} \Gamma_{k,t} \\
    %&+ \gamma \sum_{k \in S_t} p_{k,t}\bigl[\eta_{t,h} L (F_k(\bar{w}_{t, h}) - E[F_k({w_k}^*)]) \\
    %&+ \frac{1}{2\eta_{t,h}} \parallel w^{k}_{t,h} - \bar{w}_{t, h} \parallel ^2\bigr] \\
    %&- \gamma \sum_{k \in S_t} p_{k,t} (F_k(\bar{w}_{t,h}) - F_k({w}^*)) \\
    %=~&(4L{\eta_{t,h}}^2 + \gamma \eta_{t,h} L) \sum_{k \in S_t} p_{k,t} \Gamma_{k,t} \\
    %&+ \gamma(\eta_{t,h} L - 1) \sum_{k \in S_t} p_{k,t} (F_k(\bar{w}_{t,h}) - F_k({w}^*)) \\
    %&+ \frac{\gamma}{2\eta_{t,h}} \sum_{k \in S_t} p_{k,t} \parallel w^{k}_{t,h} - \bar{w}_{t, h} \parallel ^2 \\
    \le 6L{\eta_{t,h}}^2 \sum_{k \in S_t} p_{k,t} \Gamma_{k,t} + \sum_{k \in S_t} p_{k,t} \parallel w^{k}_{t,h} - \bar{w}_{t, h} \parallel ^2
\end{align*}

The inequality results from $(1) \sum_{k \in S_t} p_{k,t} (F_k(\bar{w}_{t,h}) - F_k({w}^*)) \ge 0$ and $\eta_{t,h} L - 1 \le{-\frac{3}{4}} < 0$; $(2) \Gamma_{k,t} \ge 0$ and $4L{\eta_{t,h}}^2 + \gamma \eta_{t,h} L \le{6{\eta_{t,h}}^2L}$; $(3) \frac{\gamma}{2\eta_{t,h}} \le{ 1}.$
Plug $C_1-C_2$ into $A_1$,
\vspace{-2mm}
\begin{align*} 
\vspace{-2mm}
    A_1 \le &~(1 - \mu \eta_{t,h}) \parallel \bar{w}_{t,h} - {w}^* \parallel^2 + 6L{\eta_{t,h}}^2 \sum_{k \in S_t} p_{k,t} \Gamma_{k,t} \notag \\ 
    &+ 2\sum_{k \in S_t} p_{k,t} \parallel w^{k}_{t,h} - \bar{w}_{t, h} \parallel ^2
\end{align*}
in which we have,
\vspace{-3mm}
\begin{align} 
\vspace{-2mm}
\label{eq:divergence of w}
    &E\sum_{k \in S_t}p_{k,t}\parallel w^{k}_{t,h} - \bar{w}_{t, h} \parallel ^2 %\notag\\
    %=~&E\sum_{k \in S_t} p_{k,t}\parallel  (w^{k}_{t,h} - w^{k}_{t,0}) - (\bar{w}_{t, h} - w^{k}_{t,0})\parallel^2 ]\notag \\
    %\le~&E\sum_{k \in S_t} p_{k,t} \parallel w^{k}_{t,h} - w^{k}_{t,0} \parallel^2 \notag\\
    %=~&E\sum_{k \in S_t} p_{k,t} \parallel \sum^{h-1}_{h'=0}\eta_{t, h'} g^{k}_{t,h'} \parallel^2 \notag\\
    %=~&E[\sum_{k \in S_t} p_{k,t} \parallel \sum^{h-1}_{h'=0}\eta_{t, h'} \frac{1}{|\zeta^{k}_{t,h}|}\sum_{x_{k,j}\in\zeta^{k}_{t,h}}\nabla f_k(w^{k}_{t,h'}, x_{k,j})\parallel^2]\notag\\
    %\le~&\sum_{k \in S_t} p_{k,t} (H - 1) \sum^{h-1}_{h'=0}{\eta_{t, h'}}^2 \frac{1}{{|\zeta^{k}_{t,h}|}^2}\notag\\
    %&\sum_{x_{k,j}\in\zeta^{k}_{t,h}}E\parallel \nabla f_k(w^{k}_{t,h'}, x_{k,j}) \parallel^2 \notag\notag\\
    %\le~&\sum_{k \in S_t} p_{k,t} (H-1) \sum^{h-1}_{h'=0} {\eta_{t, 0}}^2 \frac{G^2}{|\zeta^{k}_{t,h}|} \notag\\
    %\le~&\sum_{k \in S_t} p_{k,t}(H-1)^2{\eta_{t, 0}}^2\frac{G^2}{|\zeta^{k}_{t,h}|}\notag \\
    %=~&(H-1)^2{\eta_{t, 0}}^2 \sum_{k \in S_t} p_{k,t} \frac{G^2}{|\zeta^{k}_{t,h}|} \notag \\
    \leq~Q^2 (H-1)^2{\eta_{t, h}}^2 \sum_{k \in S_t} p_{k,t} \frac{G^2}{|\zeta^{k}_{t,h}|} 
\end{align}
%In the first inequality, we use $E||X - EX||^2 \le{E|| X ||^2}$ where $X = w^{k}_{t,h}- w^{k}_{t,0}$. In the second inequality, we use the following steps:
%\begin{align}
%    Var(x_k) = E({x_k}^2) - (E(x_k))^2 &\ge 0 \notag \\
%    \frac{1}{h} \sum^{h-1}_{h'=0} \parallel x_k\parallel^2 - \parallel \frac{1}{h} \sum^{h-1}_{h'=0} x_k\parallel^2 &\ge 0 \notag \\
%    \frac{1}{h^2} \parallel \sum^{h-1}_{h'=0} x_k \parallel^2 \le \frac{1}{h} \sum^{h-1}_{h'=0} \parallel x_k \parallel^2 \notag \\
%    \parallel \sum^{h-1}_{h'=0} x_k \parallel^2 \le h\sum^{h-1}_{h'=0} \parallel x_k \parallel^2 \notag.
%\end{align}
We take $\eta_{t, h'} \leq \eta_{t, 0} = \eta_{t-1, H}$ with $0 \leq h' \leq H-1$ and $0 \leq h \leq H-1$. In the third inequality, we use the assumption $E\parallel f_k(w, x_{k,j}) \parallel^2 \le{G^2}$. In the last inequality, we assume that $\eta_{t, 0} \leq Q \eta_{t, h'} $. Therefore, we bound $A_1$ as:
\vspace{-2mm}
\begin{align} \label{eq:A1}
\vspace{-4mm}
    E[A_1] \le &~(1 - \mu \eta_{t,h}) \parallel \bar{w}_{t,h} - {w}^* \parallel^2 +  6L{\eta_{t,h}}^2 \sum_{k \in S_t} p_{k,t} \Gamma_{k,t} \notag \\ 
    &+ 2Q^2 (H-1)^2{\eta_{t,h}}^2 \sum_{k \in S_t} p_{k,t} \frac{G^2}{|\zeta^{k}_{t,h}|} 
\end{align}
Assume that the loss function $f$ satisfies $E\parallel \nabla f_k(w^{k}_{t,h}, x_{k,j}) - \nabla F_k(w^{k}_{t,h})\parallel ^2 \le{\sigma^2}, \forall t, h$. Inspired by \cite{dekel2012optimal}, we bound $A_2$ as,
\vspace{-4mm}
\begin{align} 
\vspace{-2mm}
\label{eq:A2}
        E[A_2] %\notag %\\
        %=~&{\eta_{t,h}}^2 E\parallel g_{t,h} - \bar{g}_{t,h}\parallel^2 \notag \\ 
        %=~&{\eta_{t,h}}^2 E\parallel \sum_{k \in S_t} p_{k,t} \frac{1}{|\zeta^{k}_{t,h}|}\sum_{x_{k,j}\in\zeta^{k}_{t,h}}\nabla f_k(w^{k}_{t,h}) - \bar{g}_{t,h}\parallel^2\notag \\
        %=~&{\eta_{t,h}}^2 E\parallel \sum_{k \in S_t} p_{k,t} \bigl(\frac{1}{|\zeta^{k}_{t,h}|}\sum_{x_{k,j}\in\zeta^{k}_{t,h}}\nabla f_k(w^{k}_{t,h}) \notag \\
        %&- \nabla F_k(w^{k}_{t,h})\bigr)\parallel^2\notag \\
        %\le~&{\eta_{t,h}}^2 \sum_{k\in S_t} {p_{k,t}} \sum_{x_{k,j}\in\zeta^{k}_{t,h}}E\parallel \frac{1}{{|\zeta^{k}_{t,h}|}}(\nabla f_k(w^{k}_{t,h}, x_{k,j}) \notag \\
        %&- \nabla F_k(w^{k}_{t,h})) \parallel^2 \notag \\
        %=~&{\eta_{t,h}}^2 \sum_{k\in S_t} \frac{p_{k,t}}{{|\zeta^{k}_{t,h}|}^2}\sum_{x_{k,j}\in\zeta^{k}_{t,h}}E\parallel \nabla f_k(w^{k}_{t,h}, x_{k,j}) \notag \\
        %&- \nabla F_k(w^{k}_{t,h}) \parallel^2 \notag \\
        \le~{\eta_{t,h}}^2 \sum_{k \in S_t} {p_{k,t}} \frac{\sigma^2}{|\zeta^{k}_{t,h}|}
\end{align}
In the inequality, we use the Jensen inequality. \\

From Formulas \ref{eq:divergence of w}, \ref{eq:A1}, and \ref{eq:A2}, we have,
\vspace{-2mm}
\begin{align*}
\vspace{-2mm}
    &E\parallel \bar{w}_{t,h+1} - {w}^* \parallel \\
    %=~&E \parallel \bar{w}_{t,h} - \eta_{t,h} \bar{g}_{t, h} - {w}^* \parallel ^2 + E \parallel \eta_{t,h} \bar{g}_{t, h} - \eta_{t,h} g_{t, h} \parallel ^2 \notag \\
    \le~&(1 - \mu \eta_{t,h}) E \parallel \bar{w}_{t,h} - {w}^* \parallel^2 +  6L{\eta_{t,h}}^2 \sum_{k \in S_t} p_{k,t} \Gamma_{k,t} \notag \\
    &+ 2Q^2 (H-1)^2{\eta_{t,h}}^2 \sum_{k \in S_t} p_{k,t} \frac{G^2}{|\zeta^{k}_{t,h}|} + {\eta_{t,h}}^2 \sum_{k \in S_t} p_{k,t} \frac{\sigma^2}{|\zeta^{k}_{t,h}|} 
\end{align*}
We use the same batch size, i.e., $|\zeta^{k}_{t,h}| = b$, and then, we have:
\vspace{-2mm}
\begin{align*}
    &E\parallel \bar{w}_{t,h} - {w}^* \parallel \\
    %\le~&(1 - \mu \eta_{t,h}) E \parallel \bar{w}_{t,h} - {w}^* \parallel^2 +  6L{\eta_{t,h}}^2 \sum_{k \in S_t} p_{k,t} \Gamma_{k,t} \notag \\
    %&+ 2Q^2 (H-1)^2{\eta_{t,h}}^2 \sum_{k \in S_t} p_{k,t} \frac{G^2}{b} + {\eta_{t,h}}^2 \sum_{k \in S_t} p_{k,t} \frac{\sigma^2}{b} \notag \\
    %=~&(1 - \mu \eta_{t,h}) E \parallel \bar{w}_{t,h} - {w}^* \parallel^2 +  6L{\eta_{t,h}}^2 \sum_{k \in S_t} p_{k,t} \Gamma_{k,t} \notag \\
    %&+ 2Q^2 (H-1)^2{\eta_{t,h}}^2 \frac{G^2}{b} + {\eta_{t,h}}^2 \frac{\sigma^2}{b} \notag \\
    =~&(1 - \mu \eta_{t,h}) E \parallel \bar{w}_{t,h} - {w}^* \parallel^2 \notag \\
    &+ {\eta_{t,h}}^2 ( 6L \sum_{k \in S_t} p_{k,t} \Gamma_{k,t} \notag 
    + 2Q^2 (H-1)^2 \frac{G^2}{b} + \frac{\sigma^2}{b}) \notag 
\end{align*}
In addition, we exploit $\Gamma$ to represent $\max_{t\in\{1,...,T\}}\sum_{k \in S_t} p_{k,t} \Gamma_{k,t}$, and we have:
\vspace{-2mm}
\begin{align*}
    E\parallel \bar{w}_{t,h} - {w}^* \parallel 
    \leq~&(1 - \mu \eta_{t,h}) E \parallel \bar{w}_{t,h} - {w}^* \parallel^2 \\
    &+ {\eta_{t,h}}^2 ( 6L \Gamma \notag 
    + 2Q^2 (H-1)^2 \frac{G^2}{b} + \frac{\sigma^2}{b}) \notag
\end{align*}
Let $B = 6L \Gamma \notag + 2Q^2 (H-1)^2 \frac{G^2}{b} + \frac{\sigma^2}{b} \notag$, which is a constant value, and we have:
\vspace{-2mm}
\begin{align*}
    E\parallel \bar{w}_{t,h} - {w}^* \parallel &= (1 - \mu \eta_{t,h}) E \parallel \bar{w}_{t,h} - {w}^* \parallel^2 +  {\eta_{t,h}}^2 B \notag 
\end{align*}
As $1 - \mu \eta_{t,h} \geq 0$, we have $\eta_{t,h} \leq \frac{1}{\mu}$. Within the inequality of $D$, we have $\eta_{t,h} \leq \frac{1}{4L}$. Thus, $\eta_{t,h} \leq $min$\{\frac{1}{\mu}, \frac{1}{4L}\}$ for any $t$. Let us assume that $\eta_{t,h} = \frac{\beta}{(t-1)H + h + \gamma}$, with $\beta > \frac{1}{\mu}$ and $\gamma > 0$. We take $\upsilon \geq (\gamma + 1)\parallel \bar{w}_{1,1} - {w}^* \parallel$, and we have $\parallel \bar{w}_{1,1} - {w}^* \parallel \leq \frac{\upsilon}{\gamma + 1}$. Then, when $\upsilon = \frac{\beta^{m2}B}{\beta\mu - 1}$, we have $\parallel \bar{w}_{t,h} - {w}^* \parallel \leq \frac{\upsilon}{(t-1)H + h +\gamma}$ as:
\vspace{-4mm}
\begin{align*}
    \parallel \bar{w}_{t,h + 1} - {w}^* \parallel 
    %\leq~&(1 - \mu \eta_{t,h}) E \parallel \bar{w}_{t,h} - {w}^* \parallel^2 +  {\eta_{t,h}}^2 B \notag \\
    %\leq~&(1 - \frac{\mu\beta}{(t-1)H + h + \gamma}) \frac{\upsilon}{(t-1)H + h + \gamma} \\
    %&+ \frac{\beta^{m2}B}{((t-1)H + h + \gamma)^2} \notag \\
    %=~&\frac{(t-1)H + h + \gamma -1}{((t-1)H + h + \gamma)^2} \upsilon +  [\frac{\beta^{m2}B}{((t-1)H + h + \gamma)^2} \\
    %&- \frac{\beta\mu - 1}{((t-1)H + h + \gamma)^2}\upsilon]\notag \\
    %\leq~&\frac{(t-1)H + h + \gamma -1}{((t-1)H + h + \gamma)^2} \upsilon +  [\frac{\beta^{m2}B}{((t-1)H + h + \gamma)^2} \\
    %&- \frac{\beta\mu - 1}{((t-1)H + h + \gamma)^2}\frac{\beta^{m2}B}{\beta\mu - 1}]\notag \\
    %=~&\frac{(t-1)H + h + \gamma -1}{((t-1)H + h + \gamma)^2} \upsilon \notag \\
    %\leq~&\frac{(t-1)H + h + \gamma}{((t-1)H + h + \gamma)^2} \upsilon \notag \\
    \leq~\frac{\upsilon}{(t-1)H + h + \gamma} \notag
\end{align*}
By $L$-smoothness of $F$, $E\parallel F(w_{t,h}) - F({w}^*)\parallel \le \frac{L}{2}E\parallel \bar{w}_{t,h}-{w}^* \parallel^2$, we have Theorem 1:
\vspace{-2mm}
\begin{align*} 
    E\parallel F(w_{t,h}) - F({w}^*) \parallel \leq \frac{L}{2}\frac{\upsilon}{(t-1)H + h + \gamma}
\end{align*}
\end{proof}

\begin{proof}
Based on Theorem \ref{eq:eqTheorem1}, when $\upsilon = $max$\{(\gamma + 1)\parallel \bar{w}_{1,1} - {w}^* \parallel, \frac{\beta^{m2}B}{\beta\mu - 1}\}$, we have $E\parallel F(w_{T,H}) - F({w}^*) \parallel \leq \frac{L}{2}\frac{\upsilon}{(t-1)H + h + \gamma}$. When $(\gamma + 1)\parallel \bar{w}_{1,1} - {w}^* \parallel \leq \frac{\beta^{m2}B}{\beta\mu - 1}$, $\upsilon = \frac{\beta^{m2}B}{\beta\mu - 1}$. As $B$ augments with $\sum_{k \in S_t} p_{k,t} \Gamma_{k,t}$, when $\sum_{k \in S_t} p_{k,t} \Gamma_{k,t}$ is smaller, the upper bound of $E\parallel F(w_{T,H}) - F({w}^*) \parallel$ is smaller. 

We denote $\Gamma^*_t = \sum_{k \in S_t} p_{k,t} \Gamma_{k,t}$ when $p_{k,t} = \frac{\frac{n_k}{\Gamma_{k,t}}}{\sum_{k'\in S_t}\frac{n_k'}{\Gamma_{k',t}}}$ and $\Gamma_t = \sum_{k \in S_t} p_{k,t} \Gamma_{k,t}$ when $p_{k,t} = \frac{n_k}{\sum_{k'\in S_t}n_k'} $. Then, we have:
%\vspace{-3mm}
\begin{align*} \label{eq:upperBound}
    \Gamma^*_t - \Gamma_t %\\
    %=~&\sum_{k \in S_t} \frac{n_k}{\sum_{k'\in S_t}\frac{n_k'}{\Gamma_{k',t}}}  - \sum_{k \in S_t} \frac{n_k \Gamma_{k,t}}{\sum_{k'\in S_t}n_k'} \\
    %=~&\sum_{k \in S_t}\frac{n_k \sum_{k'\in S_t}n_k' - n_k \Gamma_{k,t} \sum_{k'\in S_t}\frac{n_k'}{\Gamma_{k',t}} }{(\sum_{k'\in S_t}\frac{n_k'}{\Gamma_{k',t}}) (\sum_{k'\in S_t}n_k')} \\
    =~\frac{\sum_{k \in S_t}n_k (\sum_{k'\in S_t}n_k' - \Gamma_{k,t} \sum_{k'\in S_t}\frac{n_k'}{\Gamma_{k',t}} ) }{(\sum_{k'\in S_t}\frac{n_k'}{\Gamma_{k',t}}) (\sum_{k'\in S_t}n_k')}
\end{align*}
Let us denote $\delta = \sum_{k \in S_t}n_k (\sum_{k'\in S_t}n_k' - \Gamma_{k,t} \sum_{k'\in S_t}\frac{n_k'}{\Gamma_{k',t}})$, and we have:
%\vspace{-3mm}
\begin{align*} 
\delta = \sum_{k \in S_t}\sum_{k'\in S_t}\frac{n_k n_k' (\Gamma_{k',t} - \Gamma_{k,t} )}{\Gamma_{k',t}}.
\end{align*}
When $k = k'$, we have $\Gamma_{k',t} - \Gamma_{k,t} = 0$. When $k \neq k'$, we have:
%\vspace{-2mm}
\begin{align*} 
\frac{n_k n_k' (\Gamma_{k',t} - \Gamma_{k,t} )}{\Gamma_{k',t}} + \frac{n_k' n_k (\Gamma_{k,t} - \Gamma_{k',t} )}{\Gamma_{k,t}} %\\
%=~&\frac{n_k n_k'}{\Gamma_{k',t} \Gamma_{k,t}} (\Gamma_{k,t} (\Gamma_{k',t} - \Gamma_{k,t}) + \Gamma_{k',t} (\Gamma_{k,t} - \Gamma_{k',t})) \\
%=~&\frac{n_k n_k'}{\Gamma_{k',t} \Gamma_{k,t}} (- \Gamma_{k,t}^2 + 2 \Gamma_{k',t} \Gamma_{k,t} - \Gamma_{k',t}^2 ) \\
=~&-\frac{n_k n_k'}{\Gamma_{k',t} \Gamma_{k,t}}(\Gamma_{k,t} - \Gamma_{k',t})^2 \\
%\leq ~&0
\end{align*}
%\vspace{-4mm}
Thus, we have $\delta \leq 0$ and $\Gamma^*_t - \Gamma_t \leq 0$
\end{proof}

\ma{
\subsection*{Correlation Analysis}

In this study, we investigated the correlation between JS divergence and non-IID degree. We have evidenced the significance of the positive correlation between JS divergence and non-IID degree, supporting our hypothesis that higher JS divergence is associated with increased non-IID degree.
\begin{figure}[!t]
%\vspace{-2mm}
\centering
\includegraphics[width=0.45\linewidth]{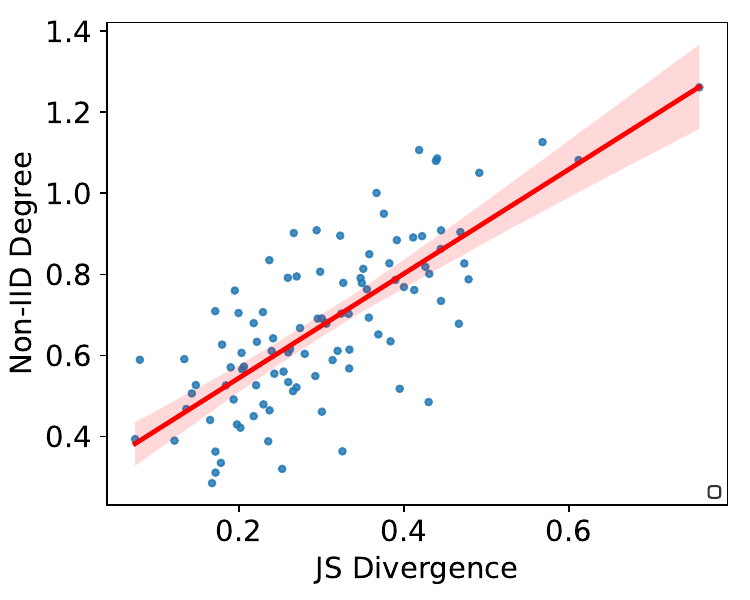}
\vspace{-4mm}
\caption{Significant positive correlations have been found between JS divergence and non-IID degree}
\label{fig:correlation_analysis}
\vspace{-6mm}
\end{figure}

\begin{table}[!t]
\centering
\caption{Performance comparison on CIFAR-10 and CIFAR-100 datasets with CNN model. }
\label{tab:cmp_new}
\begin{tabular}{c|cc|cc}
\toprule
\multirow{2}{*}{Method} & \multicolumn{2}{c|}{Cifar10} & \multicolumn{2}{c}{Cifar100} \\ 
\cmidrule(r){2-5} & Acc  & Time  & Acc  & Time \\
\midrule
\textbf{\TheName{}} (ours)& \textbf{0.749}    & \textbf{3074} & \textbf{0.423}  & \textbf{4952}\\ \hline
\textbf{\TheHetName{}} (ours)& 0.744    & 3322 & 0.420   & 5018\\ \hline
\textbf{\TheDropoutName{}} (ours)& 0.732    & 3632 & 0.412    & 5324\\ \hline
FedAS & 0.705 & 3211 & 0.415 & 5172\\ \hline
FedKTL & 0.684 & 6508 & 0.396 & 6593\\ \hline
AugFL & 0.735 & 3490 & 0.407 & 5328\\ \hline
\end{tabular}
\end{table}

Based on Formula \ref{eq:non-iid-def}, we quantify the degree of non-IID as the difference between the global and expected local losses. We exploit a linear transformation of JS divergence in Formula \ref{eq:non-iid-degree}. In order to analyze the correlation between two random variables, we calculate the Pearson correlation coefficient \cite{benesty2008importance}. This statistical measure quantifies the linear relationship between the variables, providing a value between -1 and 1. A value closer to 1 indicates a strong positive correlation, while a value closer to -1 indicates a strong negative correlation. We also assessed the statistical significance of the correlation using p-values. The analysis yields a Pearson correlation coefficient of $R^{**}=55.2\% (N = 100, p\text{-value} = 9.27 \times 10^{-19} < 0.0001$), which implies strong linear relationship between JS divergence and the non-IID degree. This strong positive correlation is visualized in Figure \ref{fig:correlation_analysis}.
}

\m{\subsection*{Comparitive Experiments}

We evaluate the performance of our proposed methods against three recent advancements as objects for comparative experiments (e.g. FedAS \cite{yang2024fedas}, FedKTL \cite{zhang2024upload}, and AugFL \cite{yue2025augfl}) using three key metrics: accuracy, training
time. Accuracy measures the proportion of correctly classified samples in the
test dataset using the trained global model. Training time is measured in seconds (s) and represents
the time required to achieve a target accuracy level, which varies by dataset. We set target accuracies of 0.72 for CIFAR-10, and 0.41 for CIFAR-100. As shown in Table \ref{tab:cmp_new}, \TheName{} consistently outperforms the three state-of-the-art, resulting in significantly higher accuracy (up to 1.93\%) and training time (up to 2.12 times shorter).}

\m{In addition, we conduct an ablation study to compare the combination of \TheName{} and AEDFL \cite{liu2024aedfl}. We consider a decentralized FL system with 100 devices with an exponential graph topology. The results are shown in Table \ref{tab:cmp_new}. The evaluation results demonstrate that the combination of AEDFL improves the accuracy by 1.9\%. Furthermore, it can further improve the efficiency of FedDHAD (up to 64.0\% ). In the future, we will investigate in the combination of FedDHAD and other approaches to achieve superior performance. }

\begin{table}[h]
\centering
\caption{Ablation study results. “Acc” represents
the convergence accuracy. “Time” refers to the training time (s) to achieve a target accuracy, i.e., 0.59 for LeNet with CIFAR10.}
\label{tab:cmp_new}
\begin{tabular}{c|cc}
\toprule
\multirow{2}{*}{Method} & \multicolumn{2}{c}{Cifar10 \& LeNet}\\ 
\cmidrule(r){2-3} & Acc  & Time  \\
\midrule
\textbf{\TheName{}} (ours)&  0.633  & 2446\\ \hline
\textbf{\TheHetName{}} (ours)&  0.621   & 2668\\ \hline
\textbf{\TheDropoutName{}} (ours)& 0.595    & 3942\\ \hline
\textbf{AEDFL+FedDHAD} & \textbf{0.645} & \textbf{880}\\ \hline
\end{tabular}
\end{table}

\ma{\subsection*{Limitations}

While \TheName{} can be exploited to efficiently train a model with distributed data, \TheName{} relies on central server for model aggregation. When the devices are connected using a decentralized topology without a central server, e.g., ring, it would be complicated to directly utilize \TheName{} for collaborative model training. In addition, \TheName{} is validated on traditional models, such as LeNet, CNN, further improvement is needed to adapt \TheName{} to other models, e.g., Transformers.}

%You can choose not to have a title for an appendix if you want by leaving the argument blank

% \section{References Section}
% You can use a bibliography generated by BibTeX as a .bbl file.
%  BibTeX documentation can be easily obtained at:
%  http://mirror.ctan.org/biblio/bibtex/contrib/doc/
%  The IEEEtran BibTeX style support page is:
%  http://www.michaelshell.org/tex/ieeetran/bibtex/
\bibliographystyle{unsrt}4
\bibliography{ref}

 

% \end{IEEEbiography}



\vfill

\end{document}

%% file: math_commands.tex
\usepackage{amsmath,amsfonts,bm}

\def\eqref#1{equation~\ref{#1}}
\def\1{\bm{1}}

\DeclareMathAlphabet{\mathsfit}{\encodingdefault}{\sfdefault}{m}{sl}
\SetMathAlphabet{\mathsfit}{bold}{\encodingdefault}{\sfdefault}{bx}{n}

\DeclareMathOperator*{\argmin}{arg\,min}